\documentclass[11pt]{article}

\usepackage[a4paper,left=28mm,right=25mm,top=30mm,bottom=30mm]{geometry}

\usepackage{epsfig}
\usepackage{wrapfig}
\usepackage{enumerate}
\usepackage{amsfonts}
\usepackage{amssymb}
\usepackage{latexsym}
\usepackage{amsmath}
\usepackage{bm} 
\usepackage{here}
\usepackage{amsthm}

\newcommand{\had}{\ensuremath{\scriptsize\begin{pmatrix} 
1 & \phantom{-}1 \\ 
1 & -1\\ 
\end{pmatrix}}}

\newcommand{\lit}{\mbox{\rm logit}\,}
\newcommand{\odr}{\mbox{\rm odr}\,}
\newcommand{\chd}{\mbox{\rm chd}\,}

\newcommand{\fracshalf}{\mbox{$\frac{1}{2}$}}
\newcommand{\fracsfourth}{\mbox{$\frac{1}{4}$}}
\newcommand{\fracseigth}{\mbox{$\frac{1}{8}$}}

\newcommand{\nn}[0]{\hspace*{.7em}}
\newcommand{\n}[0]{\hspace*{.35em}}
\newcommand{\fourl}{\nn \nn}
\newcommand{\T}{^{\mathrm{T}}}

\newcommand\cov{\mathrm{cov}}
\newcommand{\snode}{\mbox {\large 
{$\mbox{$\circ$}$}}}
\newcommand{\ful}{\mbox{$\, \frac{ \nn  \;}{ \nn }$}}
\newcommand{\fra}{\mbox{$\hspace{.05em} \frac{\nn \nn}{\nn}\!\!\!\!\! \succ \! \hspace{.25ex}$}}
\newcommand{\fla}{\mbox{$\hspace{.05em} \prec \!\!\!\!\!\frac{\nn \nn}{\nn}$}}
\newcommand{\dal}{\mbox{$  \frac{\n}{\n}\frac{\; \,}{\;}  \frac{\n}{\n}$}}

\newcommand{\ci}{\mbox{\protect{ $ \perp \hspace{-2.3ex} \perp$ }}}
\newcommand{\dep}{\pitchfork}
\newcommand{\txt}{\textstyle}

\newcommand{\E}{{\it E}}

\newcommand{\mcal}{\ensuremath{\mathcal{M}}}

\newcommand{\hcal}{\ensuremath{\mathcal{H}}}

\newcommand{\In}{\mathrm{In}}
\newcommand{\inv}{\mathrm{inv}}
\newcommand{\zer}{\mathrm{zer}}

\newtheorem{prop}{Proposition}
\newtheorem{coro}{Corollary}
\newtheorem{thm}{Theorem}
\makeatletter
\renewcommand\section{\@startsection{section}{1}{\z@}%
{-3.25ex\@plus -1ex \@minus -.2ex}{1.5ex \@plus .2ex}%
{\normalfont\large\bfseries}}
\renewcommand\subsection{\@startsection{subsection}{2}{\z@}%
{-3.25ex\@plus -1ex \@minus -.2ex}%
{1.5ex \@plus .2ex}%
{\normalfont\normalsize\bfseries}}
\renewcommand\subsubsection{\@startsection{subsubsection}{3}{\z@}%
{-3.25ex\@plus -1ex \@minus -.2ex}%
{1.5ex \@plus .2ex}%
{\normalfont\normalsize\bfseries}}
\renewcommand\paragraph{\@startsection{paragraph}{4}{\parindent}%
{3.25ex \@plus1ex \@minus .2ex}%
{-1em}%
{\normalfont\normalsize\bfseries}}
\makeatother

\usepackage{fancyhdr}
\makeatletter
\renewcommand\section{\@startsection{section}{1}{\z@}%
{-3.25ex\@plus -1ex \@minus -.2ex}{1.5ex \@plus .2ex}%
{\normalfont\large\bfseries}}
\renewcommand\subsection{\@startsection{subsection}{2}{\z@}%
{-3.25ex\@plus -1ex \@minus -.2ex}%
{1.5ex \@plus .2ex}%
{\normalfont\normalsize\bfseries}}
\renewcommand\subsubsection{\@startsection{subsubsection}{3}{\z@}%
{-3.25ex\@plus -1ex \@minus -.2ex}%
{1.5ex \@plus .2ex}%
{\normalfont\normalsize\bfseries}}
\renewcommand\paragraph{\@startsection{paragraph}{4}{\parindent}%
{3.25ex \@plus1ex \@minus .2ex}%
{-1em}%
{\normalfont\normalsize\bfseries}}

\begin{document}

\noindent{{\begin{center}{\large \bf  Generating large Ising  models with  Markov structure   via simple linear relations}\end{center}}
\noindent {\begin{center}{\bf{Nanny Wermuth}} \\
 {\em Chalmers University of Technology, Gothenburg, Sweden and Gutenberg-University, Mainz, Germany}\end{center}\n \\[-15mm]
 \noindent {\begin{center}{\bf{Giovanni Marchetti}}\\
 {\em Dipartimento di Statistica, Informatica, Applicazioni ``G. Parenti'', Florence, Italy}
\end{center}\n\\

\small{\noindent {\bf Abstract:} {\small 
 We extend the notion of a tree graph to  sequences of prime graphs
  which are 
cycles and edges and name these non-chordal graphs  hollow trees. These structures are especially attractive for palindromic Ising models,  which mimic a symmetry of joint Gaussian distributions.  We show  that for  an Ising model
 all defining independences  are captured  by zero partial correlations and conditional   correlations agree  with partial correlations within each prime graph if and only if the model is palindromic and has a hollow-tree structure.
This implies that the strength of
 dependences can be assessed locally. We use the results to find a well-fitting general Ising model with hollow-tree structure for a set of longitudinal data.
 \\[-2mm]

\small{\noindent {\bf Keywords:}  Chain graph;   Chordal graph; Chordless cycle;  Concentration  graph; Cut-set; Elimination scheme; Hollow tree; Log-linear interaction; Logit regression;  Palindromic distribution;  Prime graph;   Probabilistic graph; Quadratic exponenential distribution. \\[-6mm] 
}}
\normalsize
\subsection*{1 Introduction} 

After Maurice Bartlett had defined in 1935 log-linear interactions in trivariate contingency tables, it took almost thirty years until sets of minimal sufficient statistics for maximum-likelihood estimates  (mle) were derived  by Martin Birch in 1963, for the model without a 3-factor interaction and for a range of   possible hypotheses of conditional independence for the  discrete variables. Almost
another 10 years passed until a  general iterative-proportional fitting algorithm became available  in published form as well as relations of log-linear to the so-called   logit-regressions; Bishop (1972).

Since then, log-linear models became widely used, at the beginning helped by  a Fortran Program distributed by Leo Goodman,
called  `Everybody's Contingency Table Analysis (ECTA)'. But strong complaints emerged   that some scientists continued to use methods of data analysis based on correlations.  For instance in the preamble of a book by Goodman (1984), `correlationers'  are suggested to be in stark contrast to `crosstabbers'. The former were seen  to  rely  on linear models  and linear regressions  and the latter on contingency tables and the corresponding log-linear models, logit regressions and
odds-ratios.
One main purpose of the current paper is to identify  and  characterize  a large class of  models for binary variables  in which  particular sequences of  generating  logit  regressions  are  indeed equivalent to  sequences of linear regression.

A  theoretical discussion of linear versus log-linear models for contingency tables is  by Darroch and  Speed (1983).
Recently,   more methodological research was requested to come to a better understanding when linear relations might  be appropriate; see Hagenaars (2015).  In the 1970's in Germany, it took  data from an experiment carried out  by psychologist  Gustav Lienert,  to convince the community that correlations could sometimes be extremely misleading; for an English data description  see Wermuth (1998). The structure in Lienert's data is very close to having a  strong dependence among  three binary variables in spite of  three pairwise independences, that is when also  the  three simple  Pearson's correlation coefficients are zero.

So far, there appear  to be not  many results concerning linear relations for binary distributions. One is that directed acyclic graph models can be generated 
with linear main-effect regressions provided that all variables have mean zero, unit variance and they are symmetric; see Wermuth, Marchetti and Cox (2009). However, these models have  mainly been used  to illus\-trate  small graph-structured models and  to demonstrate  the  possible equivalence of different types of  such models. 
Relations between probabilities and correlations in Bernoulli distributions with  special types of graph, named  star graphs, see Figure 
\ref{IncGraphs3&4} for the simplest example, are by  Wermuth and Marchetti (2014) and
Wermuth, Marchetti and Zwiernik (2014).

Another result   is by  Loh and Wainright (2013). They prove  that a conditional independence constraint, represented by the missing edge in an undirected graph, leads  for binary variables to  a zero  partial correlation  given all remaining variables if the  graph is,  what we name,  a bulged  tree.  For an example of such a graph see Figure \ref{TreeTypes}.  
Unfortunately,  partial correlations may then also be zero when  they correspond to edges  present  in the graph. In that case, one cannot use zeros in the matrix 
of partial correlations to recognize  an independence  structure.

This happens  for instance, for the following two contingency  tables with four variables.  In the first,  the graph structure is a  paw  graph, see Figure \ref{IncGraphs3&4},  with all edges present for the nodes $1,2,3$  and one edge for 3,4, but the only nonzero marginal correlation is $\rho_{34}$.  In the second, no edge is missing in the  graph  since the log-linear 4-factor interaction, $\lambda_{1234}$ defined here with equation \eqref{DefInt}, is large, but nevertheless all six $\rho_{ij} $ are zero.
 $$\setlength{\arraycolsep}{.24\arraycolsep}\centering{\small
\left[\begin{array}{r r r r r r r r r r  r r r r r r r r r }

 ijk\ell\footnotemark\!:&0000 & 1000 & 0100 & 1100 &0010& 1010 & 0110 & 1110  &0001 & 1001 & 0101 & 1101 &0011& 1011 & 0111 & 1111\\
100\, \pi_{ijk\ell}\!: &4& 16& 16& 4& 4& 1& 1& 4& 1& 4& 4& 1& 16& 4& 4& 16\\
880\, \pi_{ijk\ell}\!:&100& 10& 10& 100& 10  &100 &100  &10& 10  &100 &100  &10 &100& 10& 10& 100 \end{array}\right] }
\footnotetext{\small  lower levels are here indicated by `0', also  later, in order to save space or to simplify the notation}$$

Such types of tables  may describe   sequences of bacterial genoms; see Radavi\v{c}ius, Reca\v{c}i\={u}s and \v{Z}idanavi\v{c}iut\'e (2017). For the analysis of such combinations of symmetric and antisymmetric sequences,  logit regressions may be  useful  but   linear regressions are clearly useless.
 
One might then be tempted to conclude that linear relations and hence correlations become relevant when all higher than 2-factor log-linear interactions vanish, that is when they are constrained to be zero. This is the case  in the models proposed by physicist Ernst Ising in 1925 for binary variables in effect coding, $-1,1$.
In the statistical literature,  Ising models have been  studied as binary lattice models by  Besag (1974), and as  quadratic exponential distributions for binary variables by Cox and Wermuth (1994a).  

  But another  feature is  needed also  for  Ising models  to make linear relations useful.  It is 
a  central symmetry  property of joint continuous distributions,  illustrated here for just two variables.
 \begin{center}
 \includegraphics[scale=0.18]{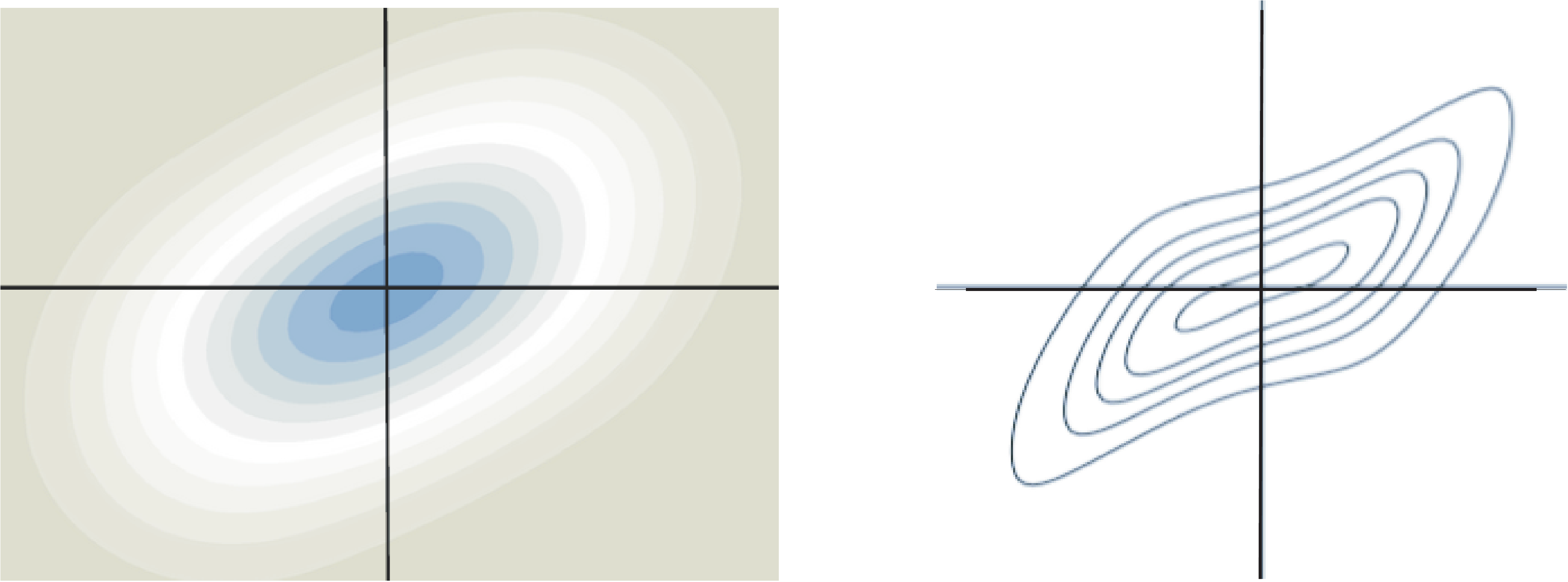}
 \end{center}
This points directly  to a main property of the following binary  distributions,  called palindromic Ising models.
 The term palindromic had been introduced in linguistics to characterise special symmetric words or sentences.
 An English
palindromic  sentence  is `step on no pets'.  

 For a single binary variable, the distribution is palindromic if it is  symmetric, that is if both levels occur with probability $1/2$. For a Bernoulli distribution  of binary variables $A_1, \dots, A_d$ with effect coding, the distribution is palindromic if  $\Pr(\omega) = \Pr(- \omega)  \text{ for all } \omega$,     
where  $\omega$ is a level combination and  $- \omega$,  is  the complement of  $\omega$ in which all signs are switched. 
For Ising models, this type of  central symmetry  is generated  with the additional constraint that each variable  taken alone is  symmetric that is  has  uniform margins.

When
probabilities are denoted by 
$\alpha, \beta$  in such bivariate distributions or  by $\alpha, \beta,  \gamma, \delta$ in such  trivariate distributions, then
the probability tables can be written as
\begin{small}
$$
\begin{tabular}{lcccc}
&& Table 1a) &&\\
\hline\\[-4mm]
$A_1$    &\hspace{-6mm}  $A_2\!: $\hspace{-5mm} &   -1   & \n1               & \text{sum}\\[1mm] \hline
\, -1    && $\alpha$         &$ \beta$ & $1/2$ \\ 
\, \n 1        &  &$ \beta$ &$ \alpha  $             & 1/2 \\[1mm] \hline
\text{sum}& & 1/2     & 1/2           & 1\\  
\hline 
\end{tabular} \hspace{10mm}
\begin{tabular}{lcccccc}
&&&Table 1b)&&&\\
\hline
    \quad         &\hspace{-6mm}  $A_3\!:$\hspace{-5mm}  &-1   &-1         & \n 1  & \n1  &           \\
  $A_1 $         &\hspace{-6mm}  $A_2\!:$\hspace{-5mm}  &-1  &\n 1         & -1 & \n1  & \text{sum}\\[1mm] \hline
\, -1    & &$\alpha $         & $\gamma  $       &$ \delta $  &$  \beta$   & 1/2    
\\ 
\, \n1         & &  $\beta $        &$ \delta  $       & $\gamma $  &$ \alpha $    & 1/2   
\\ \hline
\text{sum} &&$ \alpha+\beta$       &  $\gamma + \delta$     &$ \gamma + \delta$&
$\alpha + \beta$  & 1         \\   \hline
\end{tabular}
$$  
\end{small}
With  the assumed effect coding of the levels,  the variables have mean zero and unit variance,
so that the covariance matrix $\Sigma$ coincides with the correlation matrix, having $\sigma_{ii}=1$,
$\sigma_{ij}=\rho_{ij}$.  From Table 1a)  rewritten in terms of the simple correlation  $\rho_{12}$, one gets also
$$ \pi_{11}=(1+\rho_{12})/4,  \nn   \pi_{-11}=(1-\rho_{12})/4,  \nn   \nn  \lambda_{12}=
\fracshalf \log \{(1+\rho_{12})/(1-\rho_{12})\}=\tanh^{-1} \rho_{12},$$ 
so that $\lambda_{12}$   coincides here  with the 
$z$-transformation of $\rho_{12}$ suggested by Ronald Fisher in 1922 for a different purpose. It opens the range  $-1<\rho_{12}<1$ to the full line, $-\inf <\lambda_{12}<\inf$.

Note that for bivariate  binary distributions with skewed margins, $\rho$ is not a function of the odds-ratio,  it changes instead with the degree of skewness.
In this case, the simple  correlation coefficient is   not relevant for assessing the strength of the dependence; see Edwards (1963).

When  the column vector $  \pi$, containing the probabilities $p(\omega)$, has a lexicographic order with 
 the index of the first variable changing fastest,  the index of the last changing slowest and $H_d$ is the Hadamard matrix defined by
 the following Kronecker product:
 $$
\hcal_d = \underbrace{\had \otimes \cdots \otimes \had}_{d},
$$
then the log-linear interaction vector, $\lambda$, and  linear interaction vector, $\xi$, are respectively: 
\begin{equation}
  \lambda  =  \hcal_d^{-1} \log   \pi =(\hcal_d \log   \pi)2^{-d}, \quad \quad    \xi=\hcal_d \,\pi  . \label{DefInt}
\end{equation}
The second equality defining   $ \lambda$ is the  matrix form of an adding-and-subtracting algorithms
proposed by Frank Yates in 1937  for  effect estimation in  $2^d$ experimental designs.
This was recognized and related to more general Fourier analyses by Good (1958).

 A main result for palindromic Bernoulli distributions, proven in Marchetti and Wermuth (2016),  is their characterization by vanishing elements in $\lambda$ and   $\xi$ which involve an odd number of variables.  Such distributions get even closer to  joint distributions of standardized Gaussian variables, when they are also Ising models. In palindromic Ising models,
 the only terms of $\lambda$ permitted to be nonzero, are  the overall effect $\lambda_{\emptyset}$  and the 2-factor
 terms $\lambda_{ij}$. For independence of variables $A_i,A_j$ given all remaining variables, we use a notation  due to Dawid (1980):
\begin{equation}  \lambda_{ij}=0  \iff   i\ci j|N\setminus \{i,j\}, \label{CondInd}\end{equation}
 A set of  such independence statements  defines an  independence structure, also called  the Markov structure of the model, and shows
in  missing edges of the following type of graph.

 We treat here only simple, finite,  undirected graphs that is  at most one $ij$-edge couples distinct nodes, $i, j$,  taken from node set $N=\{1, \ldots, d\}$ having a finite $d$ and each edge present is undirected. In addition, all our starting graphs  are connected  so that each node  can be reached from another by walking along a path, that is along a sequence of distinct edges. In the following, a  `graph'  is always a member of this class of graphs.

Essential for our discussions will be subgraphs, cut-sets and prime graphs. The subgraph of a node set $a
\subseteq N$  has  as edges those present  in the graph among the nodes of $a$. A subgraph in which
all nodes are coupled is a complete graph, also called a simplex by Wagner and Halin (1962).  We use their term even though it has a different meaning in geometry. A  maximal simplex is  a clique; it  turns into an incomplete subgraph when one more node of the graph is added to it.   
A cut-set is the smallest simplex $c$ which separates disjoint subsets $a$ and $b$ of a graph, meaning  that every path between  $a$ and $b$  has a node in the complete subgraph $c$.

Prime graphs are characterised by having  no cut-set. Prime graphs are  either cliques or, if incomplete, 
each node has  some  uncoupled neighbours, we say equivalently:  each node resides  in at least two cliques.  This is evident in  chordless cycles, where each node has two uncoupled neighbours, but  is more difficult to detect  in  the other 
more complex, incomplete prime graphs. For  instance in the next to last  prime graph of Figure \ref{PrimeGraphTypes}, each of
four nodes resides in a triangle and in a hidden chordless 4-cycle.  This feature  of prime graphs
that are more complex than cliques or cycles has been  explained early on by Rose (1970, Figure \ref{PrimeGraphTypes}),  long before the concept of probabilistic graphs was  defined. 
\begin{figure}[H]
 \centering
\includegraphics[scale=.35]{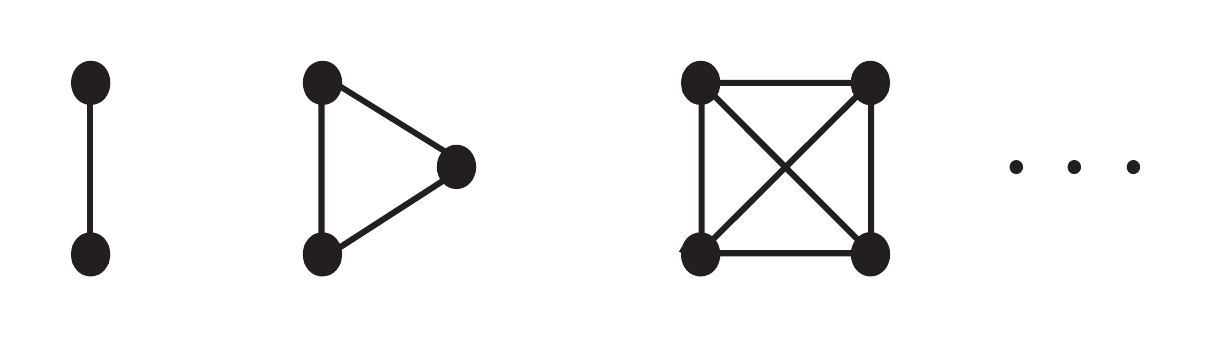}\includegraphics[scale=.35]{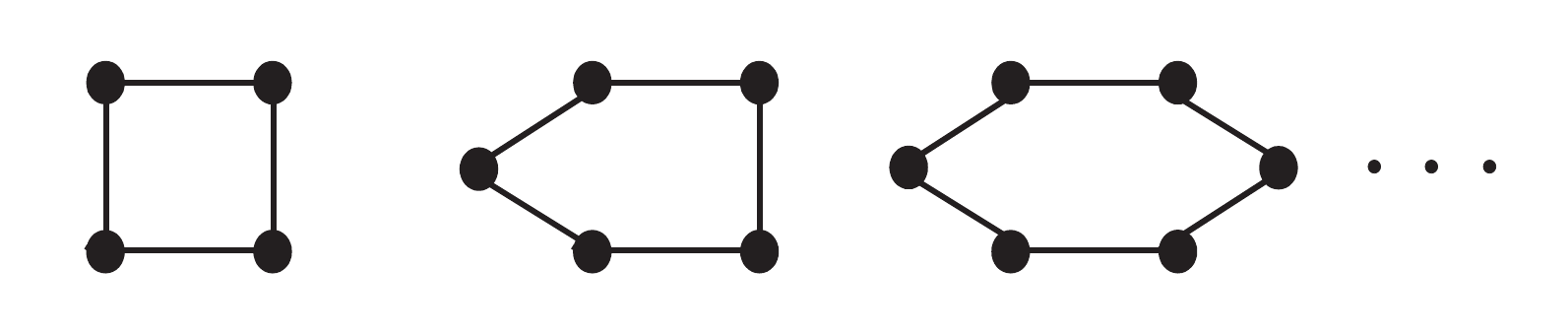}\includegraphics[scale=.35]{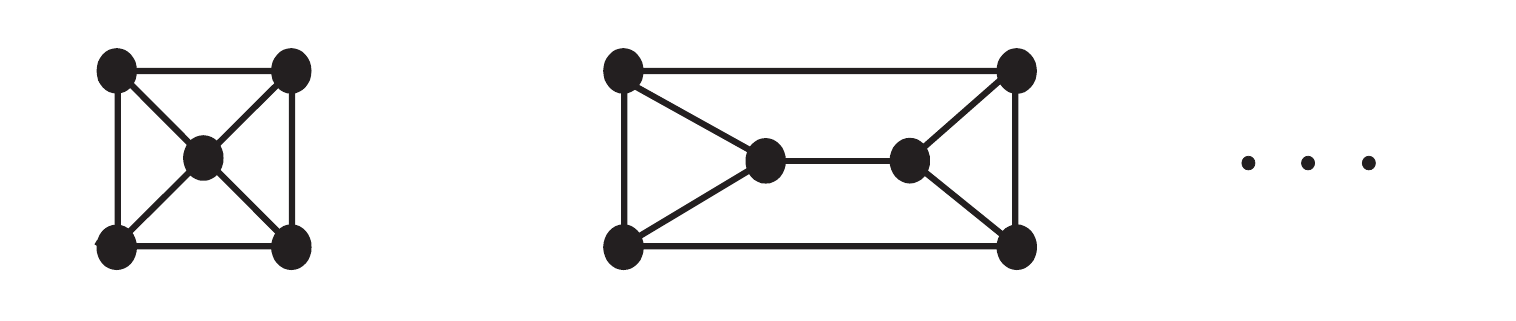}
 \caption{\small{\it the three types of prime graph in  connected  graphs, from the left: complete graphs; chordless $h$-cycles, $h>3$; other incomplete graphs, also with no node residing in a single clique}} \label{PrimeGraphTypes}
  \end{figure}
In probabilistic graphs, nodes 
represent random variables and different types of edge  permit a broad range of  conditional  dependences for
a given  independence structure,  also called  Markov structure, defined by  the  missing edges.
One  early, nearly neglected, but most important result  for probabilistic undirected graphs is  the uniqueness of their   set of prime graphs. 
\begin{thm}\!\!{\rm Wagner and Halin (1962).}\!
Every finite graph breaks into its prime graphs. \label{WagHal}\end{thm}
 Directly related to Theorem \ref{WagHal} is one  characterization of the so-called chordal  graphs: their set of prime graphs consists exclusively   of cliques.
 For less than five nodes, Figure \ref{IncGraphs3&4} shows all chordal, incomplete but connected graphs to which we get  back  later. The  prime graphs of a diamond graph are its two triangles, of a paw graph they are a triangle and an edge. In all other graphs of Figure \ref{IncGraphs3&4}, they are just edges.
\begin{figure}[H]
\centering
\includegraphics[scale =.3]{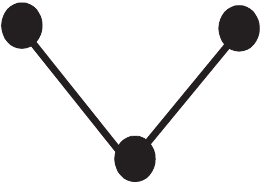} \fourl \fourl \includegraphics[scale =.3]{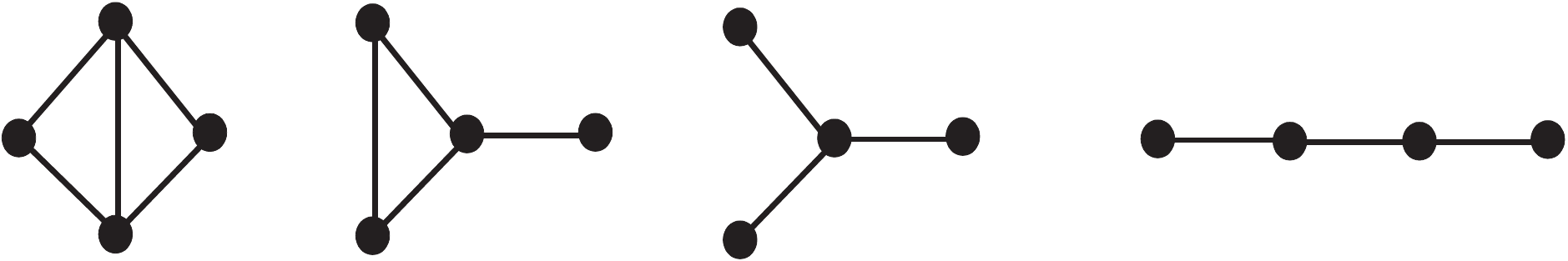}\fourl \fourl
\caption{\small{\it All  chordal, incomplete graphs in three or four nodes from the left: a transmitting {\sf V}; a diamond; a paw;  a tripod or 3-leaf star; a  single 3-edge path or 4-node Markov chain}}\label{IncGraphs3&4}
 \end{figure}\vspace{-4mm}
\noindent After the publication of Theorem \ref{WagHal},  it took  another 30 years until an  efficient algorithm became available  to find a graph's set of prime graphs; Leimer  (1993), and  until
some proofs of properties of distributions with prime graph structure were given; Mat\'u\v{s} (1994, App.). 

In computer science and statistics, a strong emphasis is until today on chordal graphs and on their so-called decomposable models; see  Tarjan and Yannakakis (1984), Darroch, Lauritzen and Speed (1980),  Thomas and Green (2009), Studen\'y and Cussens (2016).  One main reason is that the mle of  models with a chordless-cycle structure is typically fitted using an iterative procedure,  even for Gaussian  distributions; Dempster (1972); Speed and Kiiveri (1986); Dahl, Vandenberghe and  Roychowdhury (2008), Lauritzen, Uhler and Zwiernik (2017).

In this paper, we extend  traditional trees, which are chordal  graphs, to a subclass of  non-chordal graphs with attractive properties for the corresponding palindromic Ising models. 
 For this, we  distinguish first four classes  of  simple, finite,  incomplete and connected graphs using as illustration  Figure \ref{TreeTypes}. 
The traditional trees, called here also `thin trees', form one class. 
 The  other three types are extensions of thin trees.  
 When one chooses the  prime graphs of a hollow tree as the nodes of a new graph and the   cut-sets as the edges in this new graph, then  just as in a traditional tree, a `single path' connects each pair of `nodes'.

\begin{figure}[H]
\centering
\includegraphics[scale =.3]{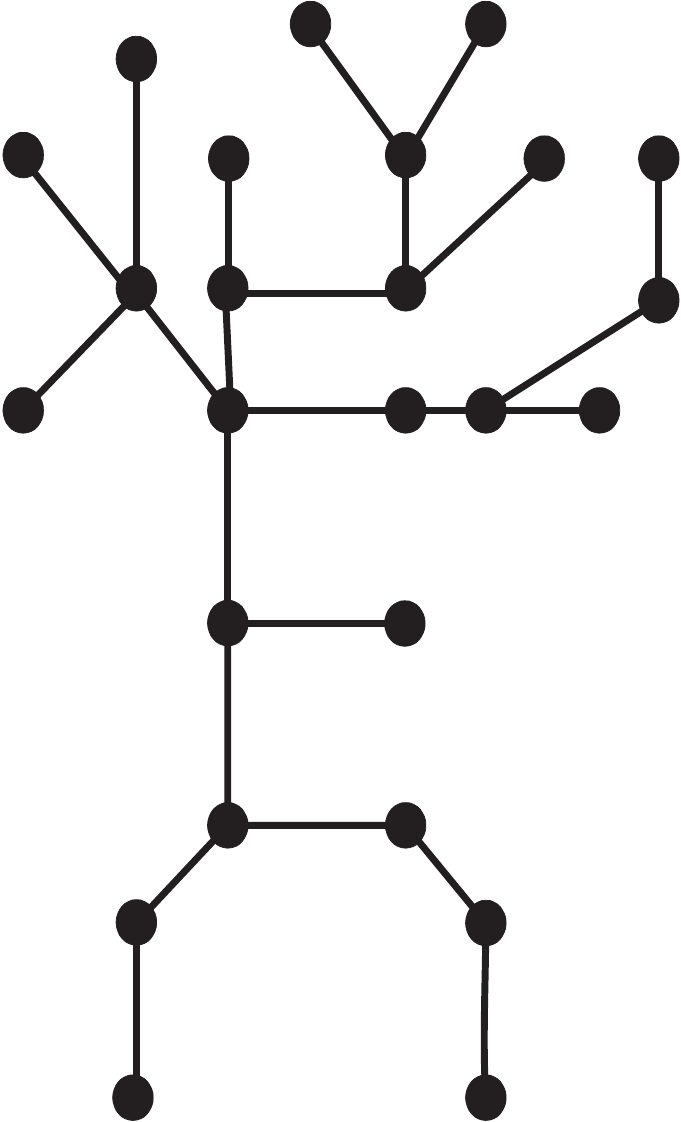} \fourl \fourl \includegraphics[scale =.3]{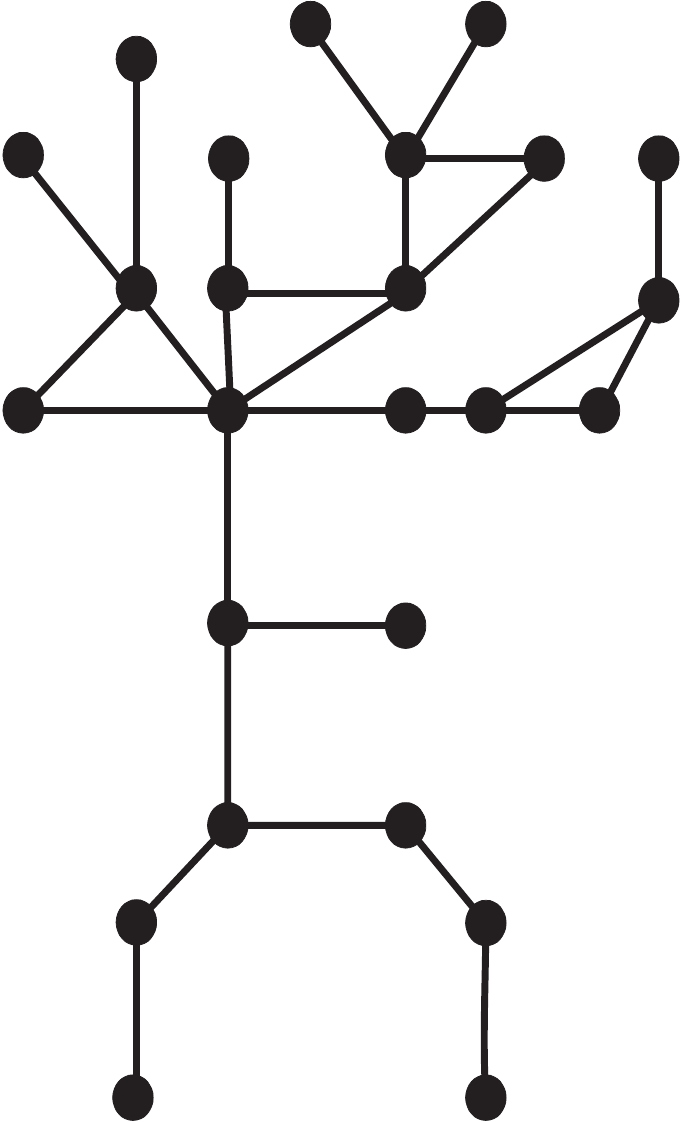}\fourl 
\fourl \includegraphics[scale =.3]{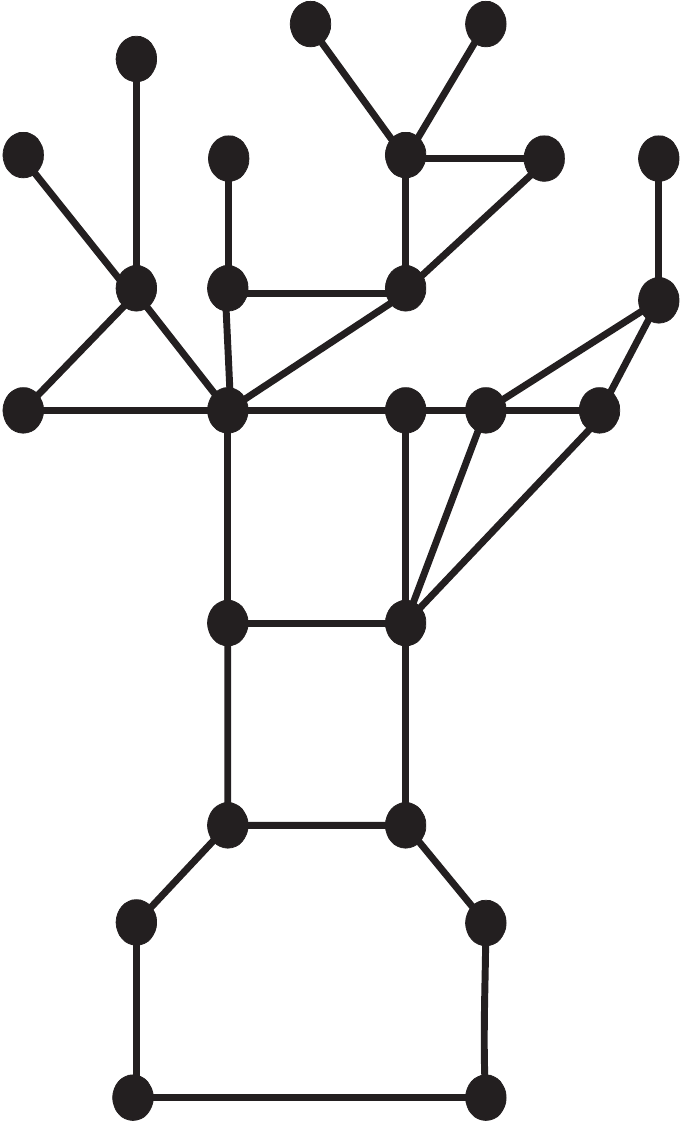} \fourl \fourl \includegraphics[scale =.3]{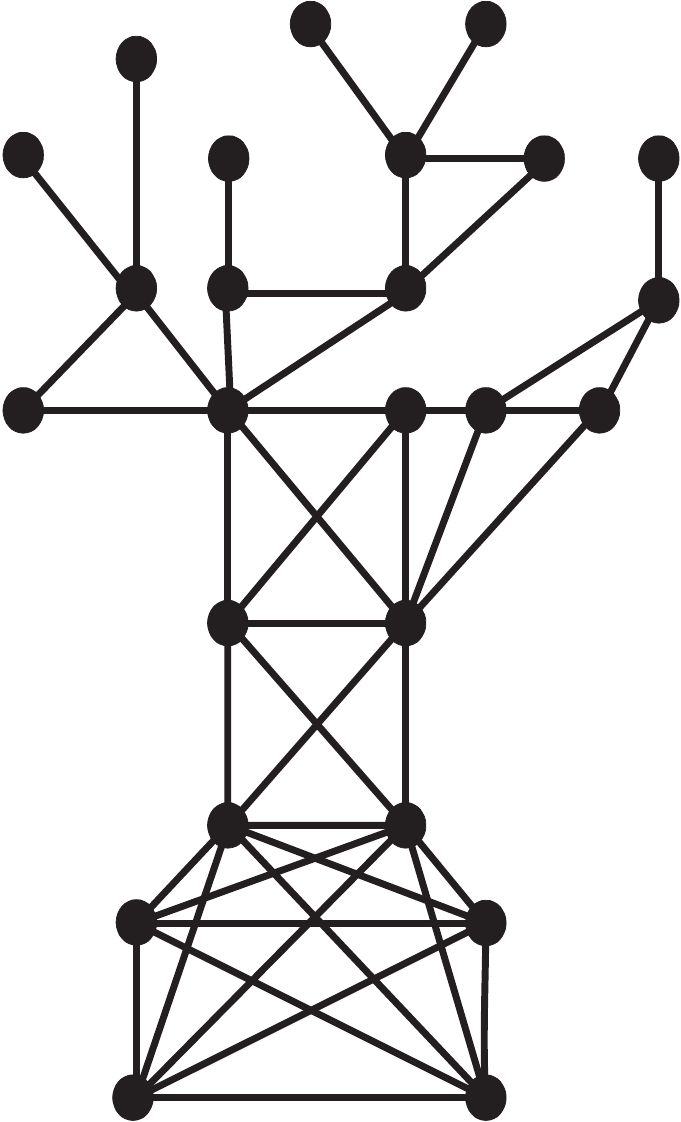}
\caption{\small{\it from the left: traditional or thin tree; bulged  tree; hollow tree; fattened tree}} \label{TreeTypes}
        \end{figure} 
 
The prime graphs of a hollow tree are exclusively edges 
or cycles -- triangles or chordless cycles -- and the cut-sets are 
 single nodes or edges. It is the only class  which may contain cycles,
the other classes consist exclusively of chordal graphs.
 Thin and bulged  trees form  subclasses of hollow trees. 
A  thin tree has only  edges as prime graphs and  single nodes  as cut-sets,
a bulged  tree consists of  edges or triangles and has cut-sets which are still just single nodes.
Finally, a chordal, fattened tree is obtained from a hollow tree by completing
each chordless $h$-cycle with all  of its  missing  chords.

The plan of the paper is to  summarise first some notation and  results for graphs in Section 2. These needed, solid justifications of the extended tree models may be skipped at a first 
reading. We 
give properties and different parametrisations for palindromic Ising models with the extended tree structures in Section 3,
stressing in particular linear parametrisations and closed form relations between parameters.  Again, the details for  special models 
are essential for our arguments but may be skipped  at a first reading. We continue with some general properties of palindromic Ising models with hollow-tree structure,   possible
transformations between  general  and  palindromic Ising models and consequences for model selection.
 In Section 4,  we illustrate the results by fitting a general Ising model to a set of longitudinal  data. 
  A discussion ends the paper. Detailed  proofs and needed previous results are given in appendices.

\subsection*{2 Summary of some graph terminology and results}

A finite, undirected graph has a set of nodes, $N=\{1,\ldots ,d\}$, for a finite  $d$. An $ij$-edge couples two distinct  nodes and turns the  nodes $i$ and  $j$ into neighbours; such nodes are  also said to be adjacent.  The graph is simple if there is at most one edge coupling each node pair and there are no loops, that is  no $ii$-edges.  Each graph can be represented by a binary symmetric matrix  with an $ij$-one if an $ij$-edge  is 
present and an $ij$-zero  if  the pair $i,j$ is uncoupled in the graph.  The matrix is called  edge matrix, if its diagonal elements are ones and  adjacency matrix if they are zeros. Only for edge matrices, $\E$, there is a matrix operator   which can capture changes  in the graph due, for instance, to marginalising; see here App. E.

In probabilistic  applications, the nodes in the  graphs represent random variables, in this paper these are binary variables $A_1, \ldots , A_d$.  If an uncoupled node pair $i,j$ means conditional independence of $A_i,A_j$ given  all remaining variables, then this type of undirected graph is also called a concentration graph and each $ij$-edge  present is drawn as a full line, $i\ful j$.

 A graph is said to be edge-minimal if every $ij$-edge indicates  a deviation from the  independence constraint for $A_i, A_j$, so that the pair  is conditionally dependent given  all remaining variables, written typically as  $i \dep j | N\setminus\{i,j\}$; see Wermuth and Sadeghi (2012). The   type of  a dependence  may be captured
by parameters  of a given probability model. In this paper, we discuss  Ising models for skewed and for symmetric binary variables, in both of which 
\begin{equation}  \lambda_{ij}\neq0  \iff   i\dep j|N\setminus \{i,j\}. \label{DefDep} \end{equation}	
Such a model is said to be generated over an edge-minimal  graph if equations \eqref{CondInd} and \eqref{DefDep} hold.

A path  is  a  sequence of distinct edges connecting its endpoint nodes.  The path is a cycle if the  endpoints coincide. A cycle has the same number of nodes as  edges and by starting at any node of a $h$-cycle and  walking  along  its $h$ edges, one returns to the starting node.  A cycle of more than three nodes is said to be chordless.

For $a, b, c$  nonempty disjoint subsets of $N$, the intersecting set $c$ of $a\cup c$ and  $b\cup c$ is a separator if every path between $a$ and $b$ has a node in $c$.
The subgraphs of separator sets may be incomplete, such as every separator in a chordless cycle, see Figure \ref{PrimeGraphTypes}, or be  complete, as all of them  are in the  
s0-called chordal graphs; for examples see Figure \ref{IncGraphs3&4}.  The smallest separating simplex of two prime graphs is a cut-set  of the graph.    In Figure \ref{IncGraphs3&4}, only the diamond graph has  an edge as a cut-set;  all other cut-sets  are single nodes. 

Prime graphs are characterized by having no cut-set. For the unique  set  of prime graphs of an undirected graph,  only maximal simplexes are used, that is the cliques.  An edge is said to be incident to $m  \subset N$ if just one of its endpoints resides in $m$. The elimination of $m$
from a graph means  to remove not only all nodes and edges within $m$ but  also all  incident to $m$. 

Leafs or outer nodes have traditionally been defined for thin trees as those single nodes which have just one  neighbour.  For chordal graphs, outer nodes   reside in a single clique. Such outer nodes have also  been called external or simplicial. The notion extends directly  to outer node-sets  which reside in a single prime graph.  After  the elimination of an  outer node set,  the  starting set of prime  graphs is reduced by just one prime graph so that, in this paper,  the remaining graph is again
 a hollow tree.

For chordal graphs,  outer single nodes can be eliminated repeatedly until only one last edge is left. This is the essence of the  Tarjan--Yannakakis algorithm which ends with a proper, also called perfect, single-node elimination scheme  or the conclusion that  the graph is not chordal. Similarly, outer node-sets can be repeatedly eliminated  from the extended trees until only one last prime graph is left.  Leimer's algorithm gives such a proper node-set elimination scheme and  the unique set of prime graphs.  In the R-package `gRbase', this algorithm is implemented;   see Dethlefsen and  H\o jsgaard (2005).  
\begin{prop}  The unique set of unlabelled prime graphs  of a connected graph remains unchanged if and only if every two prime graphs are attached at one of the graph's    cut-sets. \label{PrimeG&Cut-sets}
 \end{prop}
\begin{proof}  For a given graph,   the unique cut-sets are given by its  prime graphs and a proper  node-set elimination scheme.  
Conversely, when one attaches
a  prime graph to one which is already in a given  sequence of prime graphs,   at least one previous cut-set is destroyed and another prime graph is generated.
 \end{proof}
One  example is  a line of  $h$ edges in $h+1$ nodes. It  generates a chordless $h$-cycle when the two ending  edges
are attached  at their endpoint nodes.  Another example is a chain of four triangles. With an edge-to-edge tiling of the triangles to give the graph next to the last on the right of Figure  \ref{PrimeGraphTypes},  a new prime graph is generated. 

In 1962, Gabriel Dirac  proved that every incomplete chordal graph has at least two outer nodes; see also Blair and Peyton (1993), Lemma 3. Leimer's algorithm implies that every incomplete, connected graph has at least two outer node sets; these outer sets cannot be empty but each may also contain just a  single node.   

For thin trees, it was derived  by Camille Jordan in 1869 that in its center, there is a node or an edge  from where the paths to the outer nodes are shortest; for an efficient algorithm see Mitchell Hedetniemi,  Cockayne   and Hedetniemi (1981).
For hollow trees, the following small modification of Leimer's algorithm  finds  at its last step either a prime graph or a cut-set which is an edge:  eliminate first all outer  single nodes residing in an edge  and all outer-edge  nodes  residing in a triangle until none are left. Then,  eliminate stepwise  all outer-node sets of the remaining prime graphs. 
For instance, after  the first two steps  have been applied to the hollow tree and to the fattened tree of Figure \ref{PrimeGraphTypes}, the two tree trunks  in Figure  \ref{TreeTrunks} remain.

  \begin{figure}[H]
 \centering
\includegraphics[scale=.32]{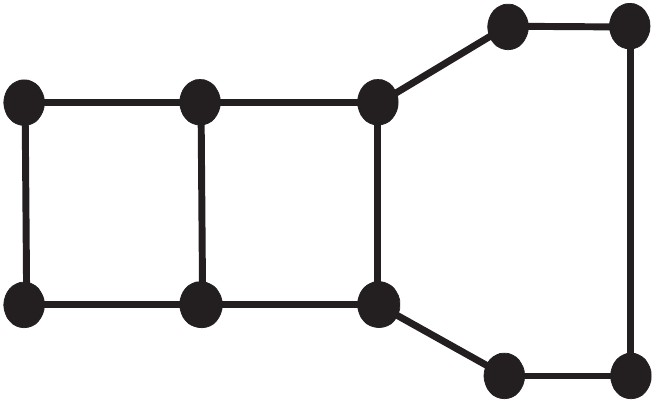}\nn \nn \nn \nn   \includegraphics[scale=.32]{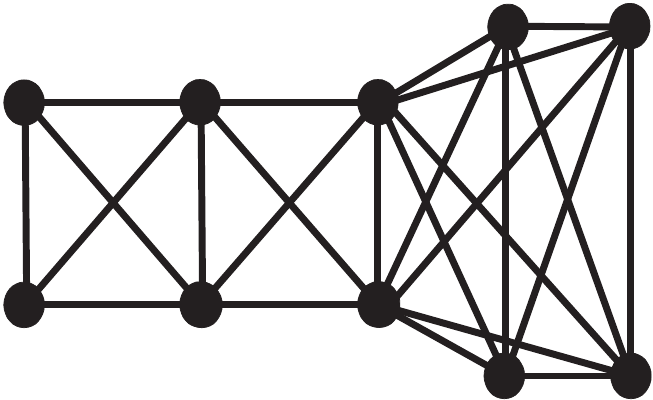}
 \caption{ \small{\it left: hollow-tree trunk to Figure \ref{TreeTypes}; right: fattened-tree trunk to Figure \ref{TreeTypes} }} \label{TreeTrunks}
  \end{figure}
 We use Figure \ref{TreeTrunks} to illustrate here how  a proper node-set elimination scheme can be  generated. If in these two graphs of Figure \ref{TreeTrunks},  the  outer node-set  on the left is eliminated first, its nodes are labelled as 1, 2. The other outer nodes in the 5-cycle  may then get    labelled as  3, 4, 5, 6. The remaining four nodes in  set  
$c=\{7,8, 9,10\}$ form then an  inner prime graph  and $c$ separates sets $a=\{1,2\}$ and $b=\{3,4,5,6\}$, that is every path from $a$ to $b$ has a node in $c$. But $c$ it is here  not  a  cut-set,
since it is   incomplete  for the graph on the left and it not a smallest simplex for the graph on the right. 

In statistics, undirected graphs were the first studied types of the  broad class of probabilistic graphs. For many probabilistic graphs one knows now when they are Markov equivalent, that is when they define the same independence structure, just  in a different
way;  see here App. F. 
One may for instance orient a hollow tree, that is change some of its full lines to arrows, by  any proper  node-set  elimination scheme and obtain a special type of chain graph,  one of those introduced by Lauritzen and Wermuth (1989) and Frydenberg (1990),  called and studied later as discrete LWF chain graphs; see Drton (2009).  
 
 Thus, the 
 hollow-tree trunk of Figure \ref{TreeTrunks}
may for instance be oriented  as in Figure \ref{OrientedTrunks}. These orientations mimic one important property of thin trees: each prime graph can be chosen as the undirected past, that is as the one from where arrows are starting.
\begin{figure}[H]
\begin{center}
\includegraphics[scale=.32]{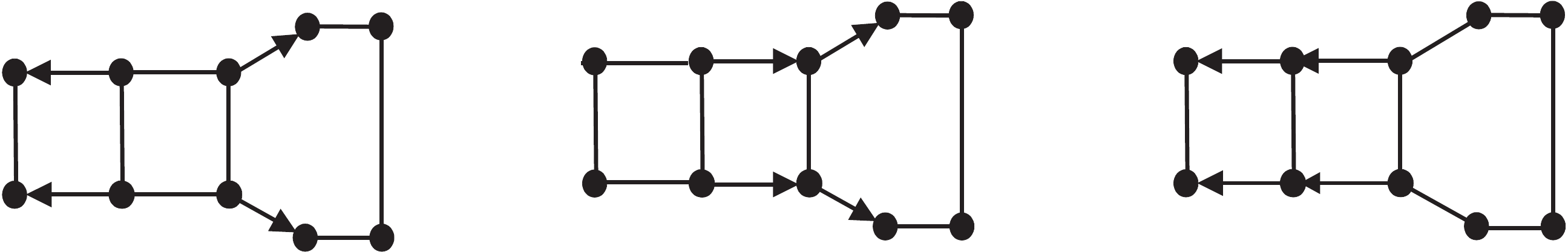} 
 \caption{ \small{\it  Markov equivalent orientations of the hollow tree  in Figure \ref{TreeTrunks} with undirected cycles, left: the 4-cycle in the middle; middle: the 4-cycle
 on the left; right: the 6-cycle on the right}} \label{OrientedTrunks} \vspace{-5mm}
\end{center}
\end{figure}

Factorisations of  a distribution such as in equation \eqref{FactMarg} below relate to elimination schemes: the outer-node sets 
correspond  to joint responses, each cut-set to a set of regressors.

\begin{prop}   In a  hollow tree and in its fattened tree,  one obtains  -- after eliminating in sequence only outer node sets -- a subset $a\cup b \cup c$
where  $a$ and $b$  are the outer nodes of
two prime graphs and the smallest separator $c$ is their  cut-set or their  inner prime graph. \label{MargPrimeG}  
 \end{prop}

\begin{proof} By the uniqueness of the set of prime graphs and  the existence 
of   a proper  node-set elimination scheme for all nodes, only the edges within and incident to $N\setminus\{a\cup b\cup c\}$ are removed,  hence no edge  is added and none is removed in the subgraph of $a\cup b\cup c$.
\end{proof}

For  Ising models,  this is exploited in Section 3  by using  covering  models, which are chordal, and reduced models,  which have cycles;  see Cox and Wermuth (1990) for these notions.  Thereby, we build heavily on  the existence of  proper  node-set elimination schemes.

\subsection*{3 On the relevance  of  linear relations for Ising models}

\subsubsection*{\small 3.1 Measures of dependence in trivariate palindromic  models}

In the palindromic model for three binary variables, 
the  3-factor interaction is missing, hence   it is an Ising model with
symmetric marginal distributions  and it has  joint symmetries.

To relate different measures of dependence later, we take Table 1b) rewritten as
\begin{center}
\begin{small}
\renewcommand*{\arraystretch}{0.8}
\begin{tabular}{lccccccc}
    \hline\\[-3mm]
    \quad         &\hspace{-6mm}  $A_3\!:$\hspace{-5mm}  &-1   &-1   &      & \n 1  & \n1  &           \\[1mm]
  $A_1 $         &\hspace{-6mm}  $A_2\!:$\hspace{-5mm}  &-1  &\n 1  &sum       & -1 & \n1  &sum\\[1mm] \hline\\[-3mm]
-1    & &$\alpha $         & $\gamma  $   &   $\fracsfourth(1+\rho_{13})$ &$ \delta $  &$  \beta$   &    $\fracsfourth(1-\rho_{13}) $
\\[1mm] 
\, 1         & &  $\beta $        &$ \delta  $   &  $\fracsfourth(1-\rho_{13})$  & $\gamma $  &$ \alpha $    &  $\fracsfourth(1+\rho_{13})  $
\\[1mm] \hline\\[-3mm]
\text{sum} &&$\fracsfourth(1+\rho_{23})\nn$      & $\fracsfourth(1-\rho_{23})$\nn   & \fracshalf&$\fracsfourth(1-\rho_{23})\nn $&
$\fracsfourth(1+\rho_{23})$  \nn& \fracshalf        \\[1mm]   \hline
\end{tabular}
\end{small}
\end{center}

For this table, we define  at level  $-1$ of $A_3$ several   standard and widely used  measures of dependence as well as further measure, $\tau$, that we name the  hypetan interaction
since it is the hype(rbolic) tan(gent) of $\lambda$. It had been introduced and studied for $2 \times 2$ tables under the name `coefficient of colligation' by Yudney Yule in 1912. The difference of the
conditional probabilities for level $1$, taken  at levels $1$ and $-1$ of another variable,  is said to be the chance difference for success. It is used
almost exclusively in the  current literature on  causal modelling.

Each of these measures remains unchanged when moving to  level 1 of $A_3$ and  
 $1\ci 2|3$ is reflected in a zero value
 of each measure except for the odds-ratio, which then equals one.
  \begin{eqnarray*}\small  \label{measdep} 
  \text{odds-ratio or cross-product ratio:} \n  \odr\!_{12|3}&\!=\!&  (\alpha \delta)/ (\beta \gamma),\nonumber\\ 
 \text{chance difference for success:}  \n \chd\!_{1|2.3}&\!=\!&  16(    \alpha   \delta - \beta \gamma) / (1-\rho_{23}^2)  \nonumber   ,\\
\text{conditional correlation:}\n \rho_{12|3}&\!=\!&16(\alpha \delta - \beta \gamma)/\{(1-\rho_{13}^2)(1-\rho_{23}^2)\}^{-1/2} \nonumber,\\
\text{log-linear interaction:}\n\lambda_{12}&\!=\!& (\log \odr\!_{12|3})/4\nonumber,\\
\text{hypetan interaction:}\n \tau_{12} &\!=\!& (\odr\!_{12|3}^{1/2}-1)/(\odr\!_{12|3}^{1/2}+1)=\tanh(\lambda_{12})\nonumber.
 \end{eqnarray*}

As explained before, linear relations may in general not be relevant for contingency tables.
But here, there is an invertible  relation of the conditional probabilities to the  simple  correlations in terms of an anti-diagonal
matrix, a Hankel matrix formed by $(-1, -1,  \; 0)$ and $(0, -1, -1)$. With $\alpha=\fracshalf-(\beta+\gamma +\delta)$, we have:
\begin{equation} \begin{pmatrix} \rho_{12}\\ \rho_{13}\\ \rho_{23}\end{pmatrix} =\begin{pmatrix} -1 & -1&\nn 0\\ -1 &\nn 0 & -1\\ \nn 0& -1 &-1\end{pmatrix} (4 \begin{pmatrix} \beta\\ \gamma\\ \delta
\end{pmatrix}-1/2).  \label{rhopi}
 \end{equation}

By applying equation \eqref{DefInt}, one sees directly that there is no log-linear and no linear 3-factor interaction.     Furthermore, as stated  next in equations \eqref{parregc} and \eqref{parcor}, chance differences coincide with linear regression coefficients, and conditional correlations with partial correlations.

For more general palindromic Ising models, linear regression coefficients  are based on the covariance matrix $\Sigma$ which coincides with the correlation matrix. We write such a  linear regression coefficient  by indicating the response before  a conditioning sign `$|$', the regressor after `$|$' and
the remaining regressors of the given response  after `.'

For  the trivariate palindromic Ising model, we have
$  16(\alpha \gamma -\beta \delta)=(\rho_{12}-\rho_{13}\rho_{23}) $ so that
\begin{equation}\beta_{1|2.3}= (\rho_{12}-\rho_{13}\rho_{23})/(1-\rho_{23}^2)  \text{ coincides with  } \chd\!_{1|2.3}, 
\label{parregc} \end{equation}
\begin{equation} \rho_{12.3}=(\rho_{12}-\rho_{13}\rho_{23})/\{(1-\rho_{13}^2)(1-\rho_{23}^2)\}^{-1/2} \text{ coincides with }  \rho_{12|3}. \label{parcor}\end{equation}

This makes it possible to study in the following 
changes in dependences for palindromic Ising models after marginalising or conditioning via  well-known recursive, linear 
relations for  Gaussian distributions; for  general proofs when different equal measures remain unchanged after marginalising, that is when they are collapsible, see
Xie, Ma and Geng (2008).
 
Recursive relations have been derived for linear regression coefficients by Bill Cochran in 1938, for covariances by Ted Anderson
in 1958 and for concentrations by Arthur Dempster in 1969. For discussions of all three in connection  with graphs   see Wermuth and Cox (1998), Marchetti and Wermuth (2009), Wermuth (2015); for simple proofs of matrix forms see Wiedenbeck and Wermuth (2010). 
$${ \beta_{1|2}=\beta_{1|2.3}+ \beta_{1|3.2}\sigma_{23}/\sigma_{22}, \nn \sigma_{12|3}=\sigma_{12}-\sigma_{13}\sigma_{23}/\sigma_{33},\nn \sigma^{12.3}=\sigma^{12}-\sigma^{13}\sigma^{23}/\sigma^{33}},\label{RecLin}$$
where for covariances, `$|$'  means again conditioning, while  for concentrations, `.' indicates marginalising. 
For  more than three variables, an additional conditioning set, $c$,  is added to the  regression coefficients and an additional marginalising
set, $m$, to the concentrations in each term  of the above equations; as for instance for equation \eqref{PcorRegcCon} below.

While  the relevance of such linear measures for palindromic hollow trees  will be shown later,
the recursive relations for $A_1, A_2, A_3$  give for  trivariate distributions  directly: 
\begin{eqnarray*}
\beta_{1|2}=\beta_{1|2.3}& \iff &{\scriptsize 1}\ci {\scriptsize 3|2}\text{ or } {\scriptsize 2}\ci {\scriptsize 3},  \\[1mm]
  \sigma_{12|3}=\sigma_{12}& \iff &  {\scriptsize 1}\ci {\scriptsize 3}\text{ or } 
{\scriptsize 2}\ci {\scriptsize 3}, \\[1mm]
\sigma^{12.3}=\sigma^{12} &\iff& {\scriptsize 1}\ci {\scriptsize 3|2}\text{ or } {\scriptsize 2}\ci {\scriptsize 3|1},
\end{eqnarray*}
expressing  conditions for collapsibility in terms of  independences.

There is also a recursive relation for  $\tau_{ij}$, which  after marginalising over $A_k$ turns into the simple correlation $\tau_{ij\backslash k}=\rho_{ij}$ of 
variables $A_i,A_j$, as proven for equation \eqref{Marg23Tri} in App. A:
\begin{equation} \rho_{ij}= \tau_{ij\backslash k}=(\tau_{ij}+\tau_{ik}\tau_{jk})/\text{const., with:   const.}=1+\tau_{ij}\tau_{ik}\tau_{jk}. \label{rectau}\end{equation}
Equation \eqref{rectau} leads directly to   conditions for collapsibility for $\tau$s, $\lambda$s and the  odds-ratios with:
$$  \tau_{12\backslash 3}=\tau_{12}  \iff {\scriptsize 1}\ci {\scriptsize 3|2}\text{ or } {\scriptsize 2}\ci {\scriptsize 3|1} \,.$$

 For three variables,   equations \eqref{DefInt} show that there is no 3-factor interaction for  the log-linear and for  the linear interaction parameters.
That   linear-regression coefficients are positive multiples of  $\tau$s is derived in App. B and shown next:
\begin{equation} \label{BetaTauLam}
\beta_{1|2.3}=\tau_{12}(1-\tau_{13}^2)/(1-\tau_{12}^2\tau_{13}^2), \nn  \n  \beta_{1|3.2}=\tau_{13}(1-\tau_{12}^2)/(1-\tau_{12}^2\tau_{13}^2), 
 \end{equation}
Since for instance   $ \rho_{12.3}$ is a  a positive multiple of $\beta_{1|2.3}$, it follows that also partial correlations  are positive multiples of 
the hypetan interaction $\tau_{12}$.

By the  definition of Pearson's correlation coefficient for binary variables, see here equation \eqref{Defccor}, the probability  $\pi_{11|k=1}$  coincides with the probability when all signs are switched and  is 
\begin{equation} \pi_{11|k=1}=2^{-3}\{{\rm const.\,} \rho_{12|3}+ (1+\rho_{13}k)(1+\rho_{23}k)\}, \nn {\rm const.}=\{(1-\rho_{13}^2)(1-\rho_{23}^2) \}^{1/2}, \label{ParcortoProb}
\end{equation}
and, since  $\rho_{ij|k}=\rho_{ij.k}$,  all probabilities are  also given by using the  partial correlation matrix. \\[-7mm]
\subsubsection*{\small 3.2 Palindromic Ising models with no independences and  with  diamond-graph structure}

For  
$d$ binary random variables $A_1, \dots, A_d$ taking values  $-1, +1$, the  probability distribution $\Pr(\omega) = \Pr(A_1=\omega_1, \ldots, A_d=\omega_d)$   is a palindromic Ising model if
\begin{equation} \label{DefPalIsLam}
\log \Pr(\omega)= \lambda_{\emptyset} + \,\txt \sum_{s < t} \lambda_{st}\, \omega_s \,\omega_t, \n -\infty <\lambda_{st}< \infty \,.
\end{equation}  
In such models with concentration graph structure,   a missing  edge for $i,j$  in the graph indicates the   conditional independence  of equation \eqref{CondInd} and an $ij$-edge present means the conditional dependence of equation \eqref{DefDep}. The interactions $\lambda_{ij}$,  $\xi_{ij}$ result as in equation \eqref{DefInt}.

By   equation  \eqref{DefPalIsLam}, the log-linear interactions, $\lambda_{ij}$,  and hence the 
$\tau_{ij}=\tanh \lambda_{ij}$,  are  constant at all level combinations of the remaining variables. Similar to $d=3$:
$$\lambda_{ij}=\fracshalf  \log(1+\tau_{ij})/(1-\tau_{ij}), \n \nn \tau_{ij}=(\odr\!_{ij|N\setminus\{i,j\}}^{1/2}-1)/(\odr\!_{ij|N\setminus\{i,j\}}^{1/2} +1).$$
But for $d>3$,  this does  in general not imply a linear model  with only 2-factor linear interactions.
In particular,  for $\rho_{ij|N \setminus \{i,j\}} \neq 0$  varying 
 at several level combinations of the remaining variables, these conditional  correlations cannot equal the single value of
the  partial correlation, $\rho_{ij.N \setminus \{i,j\}}$; but see  Proposition \ref{CEqLogLin} as well as Theorem \ref{IsGaussConc} below.

As an illustration, we choose the palindromic Ising model 
with the diamond graph structure shown in Figure \ref{IncGraphs3&4}, the smallest covering model for a chordless 4-chain.  We  denote the variables by $A,B,C,D$, their levels by $i,j,k,\ell$ and $p(\omega)=\pi_{ijk\ell}$.  Conditional probabilities, for instance for $A$ given  the level combinations of  $B,C,D$, are then  $\pi_{i|jk\ell}= \pi_{ijk\ell}/\txt\sum_i \pi_{ijk\ell}$.
We label the outer nodes of the diamond graph as $1, 2$ and the cut-set nodes as $3,4$ so that the missing edge is for nodes $1,2$. An example is displayed  in the next  table:

$$
\centering
{\small
\left[\begin{array}{r r r r r r r r r r  r r r r r r r r r } 
  ijk\ell \text{ with }l=0:&0000 & 1000 & 0100 & 1100 &0010& 1010 & 0110 & 1110  \\
302\, \pi_{ijk\ell}\,: &100 & 20& 5 & 1& 12& 8& 3& 2\\
\end{array}\right]	}
$$
The following matrix on the left contains  $\lambda_{ij}$ in the lower and  $\tau_{ij}$ in the upper-triangular part, the matrix  on the right  has $\rho_{ij}$ in the lower and $\rho_{ij.k\ell}$ in the upper-triangular part.  

{\small
$$    
\begin{array}{cccc}\hline
                           0 & 0 &  0.292&0.465 \\
                           0&0& 0.382&0.799\\
                           0.301 &0.402& 0& 0.342 \\ 
                            0.504 &  1.100  & 0.357& 0\\  \hline
                            \end{array} \quad  \fourl
                            \begin{array}{cccc} \hline
                            1 &  0 &  0.210& 0.265 \\
                           0.523 & 1 & 0.213&0.714\\
                     0.523   &  0.656  &   1& 0.206\\ 
                     0.589 & 0.854&   0.669& 1\\
                     \hline
                         \end{array} \quad
                       $$        
}   

The linear regression of either $A_3$ or $A_4$ on all remaining variables induces a linear 4-factor-interaction
in the joint distribution, so that  the nonzero conditional correlations vary with the level combinations of $A_1, A_2$
and  hence, can for these responses not  coincide with the single value of the corresponding  partial correlation in the joint distribution of all four variables.

But the 4-factor 
linear interaction is irrelevant,  when this  palindromic  Ising model is generated  by  three logit equations,  for which  the  nodes are ordered as 
 1,2,3,4,  and each response has just two regressors. To simplify notation, we change again level $-1$  to $0$.
 \begin{eqnarray} \label{logittri}
\lit(\pi_{1|k\ell})=\log \pi_{1+k\ell}-\log \pi_{0+k\ell} &=& 2 ( \lambda_{13}\,k+\lambda_{14}\,\ell )\nonumber \fourl \text{ for response  $A$},\\
\lit(\pi_{1|k\ell})=\log \pi_{+1k\ell}- \log \pi_{+0k\ell}&=&2(\lambda_{23}\,k +\lambda_{24}\,\ell)  \nn \nn \text{ for response  $B$}, \\
\lit(\pi_{1|\ell})=\log \pi_{++1\ell}- \log \pi_{++0\ell}&=& 2(\lambda_{34} \, l)\nonumber  \fourl \fourl  \fourl \n \, \text{ for  response  $C$,}
\end{eqnarray} 
where, for instance,   for response $A$: $ \pi_{1+k\ell}\!=\!\sum_j \pi_{1jk\ell}$,  for response $C$: $\pi_{++1\ell}\!=\!\sum_{ij} \pi_{ij1\ell}$. Only   term $2 \lambda_{12} \, j$ is missing compared to  a  model without any independence constraint. 

The term `logit' had been introduced for  ratios of rates and in logit-regressions rates depend  only on binary variables. For the above ratios of conditional probabilities, the  terms conditioned upon  cancel so that only the probabilities in the two 3-node triangles, for subsets $\{1,3,4\}$, $\{2,3,4\}$, and within the cut-set of the graph, $ \{3,4\}$, are used to generate the model.

A strongly condensed, common  notation for densities and probabilities reflects  both:  independence structures  and  proper  elimination schemes.  For instance  for the diamond graph and for the  $4$-node Markov chain in  Figure \ref{IncGraphs3&4}, respectively, we can write:
$$f_N=f_{1|34} f_{2|34} f_{3|4}f_4, \fourl  f_N=f_{1|2} f_{2|3} f_{3|4} f_4.$$

This type of  factorisation  exists if and only  if the model has a  chordal graph structure. The reason is that for chordal graphs and only for those, every proper single-node elimination scheme results in a directed acyclic graph  in which none of the previous {\sf V}s,  $\snode \ful \snode  \ful \snode$,  corresponds  to a sink {\sf V}, that is to a subgraph
of the type  $\snode \fra \snode  \fla \snode$;  see  here Theorem 6 in App. F  and  for its proof see Wermuth and Sadeghi (2012, Theorem 1). 

The diamond graph distribution of equation \eqref{logittri} may, by its factorisation and by the inverse equation \eqref{rhopi}, equivalently  be generated with the following linear regressions:
\begin{eqnarray} \label{lintri}
\pi_{i|k\ell}&=& \fracshalf\{1+ \beta_{1|3.4}\,ik+\beta_{1|4.3}\, i\ell)\}\nonumber ,\\
\pi_{j|k\ell}&=& \fracshalf\{1+ \beta_{2|3.4}\,ik+\beta_{2|4.3}\,i\ell)\}, \\
\pi_{k|\ell}&=& \fracshalf\{1+ \beta_{3|4}\,k\ell\} \nonumber,\\
\pi_{\ell}&=& \fracshalf \nonumber \, .
\end{eqnarray}

For each of $ f_{134},  f_{234}$,  the conditional correlations and the partial correlations  coincide by definition, see the previous section.
Hence, the strength of the dependences can be judged from the partial correlations for $A,C,D$ and $B,C,D$, while the independence 
for pair (1,2) shows in  the zero partial correlation  $\rho_{12.34}$, the standardised version of the concentration $\sigma^{12}$.

Induced simple correlations  result for $a\ci b|c$ -- as in Gaussian distributions -- from requiring $\Sigma_{ab|c}=0$, here  for $a=\{1\}$, $b=\{2\}$, and $c=\{3,4\}$, with 
\begin{equation}\Sigma_{ab}^*=\Pi_{a|c}^{\n}\Sigma_{cb}^{\n}, \nn \n  \Pi_{a|c}=\Sigma_{ac}^{\n}\Sigma_{cc}^{-1} \label{InducCov}  \end{equation}
where $\Pi_{a|c}$ denotes the matrix of linear regression coefficients for response nodes $a$ and regressor nodes $b$.
More explicitly, after indicating a transposed vector by $(\n)\T$, equation \eqref{InducCov} becomes for the diamond graph, labeled as above,
$ \rho_{12}^*= (\beta_{1|3.4}, \; \beta_{1|4.3}) (\rho_{32} \; \rho_{42})\T.
$

 It follows from equation \eqref{DefInt}, or more directly from  
equation (2.4) in Wermuth, Marchetti and Cox (2009),  that the 4-factor linear  interaction,  $\xi_{ijk\ell}$,  for $A,B,C,D$ with  $e=(j, k,\ell)$ is:
\begin{equation}\xi_{ijk\ell}=\Pi_{i|e} (\rho_{k\ell}, \rho_{j\ell}, \rho_{jk})\T.   \label{InducInt}  \end{equation}
This gives for the  diamond graph,  labeled as above,
$  \xi_{1234}=(\beta_{1|3.4}, \; \beta_{1|4.3}) (\rho_{24}, \;\rho_{23})\T$. For response nodes  3 or 4, this interaction implies in regressions on all remaining nodes  varying  linear coefficients even though  
the logit-regression coefficients are  constant.

In trivariate palindromic Ising models and in triangles as prime graphs, linear regression coefficients  are constant  by definition,  but  by 
Proposition \ref{CEqLogLin},  this constancy is impossible for a node chosen as response when it still has  more than two neighbours.
\begin{prop}   \label{CEqLogLin} For  palindromic Ising  models,  a main-effect logit-regression  of   response $A_i$ is equivalent to a main-effect linear regression of $A_i$  if and only if  $A_i$ has at most two regressors. 
\end{prop}
\begin{proof} The claim holds for three variables  with equation \eqref{BetaTauLam} given  the equivalence of the parametrizations in terms of $\lambda$s and $\tau$s, in equations \eqref{DefPalIsLam} and \eqref{DefPalIsTau}. For a response with three or more regressors, a  4-factor linear interaction is introduced  
in the joint distribution; see  equation \eqref{InducInt}.
It implies a 3-factor linear interaction in the conditional distribution of this response given more than two  regressors, hence destroys 
the equivalence. \end{proof}
Proposition  \ref{CEqLogLin}
 explains, for an inner node chosen as  a response,  when main-effect logit regressions can be replaced by  main-effect linear regressions.
This holds for instance, for   the  chordal graphs on the left and on the right of Figure \ref{IncGraphs3&4}, but not for the other three and  stresses  the importance of proper elimination schemes  in order to obtain relevant  linear regressions.

Relations of  regression coefficients and partial correlations to concentrations, follow most directly with a matrix operator which extends the sweep operator and is  defined here in  App. E.  For a partitioning of a node set as $\{1, 2, c,m\}$, one has
\begin{equation}\rho_{12.c}=- \sigma^{12.m}/\{\sigma^{11.m}\sigma^{22.m}\}^{-1/2},\fourl \beta_{1|2.c}=-\sigma^{12.m}/\sigma^{22.m}.\label{PcorRegcCon}
\end{equation}

By contrast, within each $2\times 2$ subtable, the conditional correlation is just a correlation coefficient for  two binary variables.  Here, we  replace again level $-1$ by $0$, let $e=  N\setminus\{m,1,2\}$ denote a fixed level combination of the remaining variables and e.g.  $\pi_{+1e}=\txt{\sum}_i \pi_{i1e}$ 
\begin{equation} \rho_{12|e}=(\pi_{11e}-\pi_{+1e}\pi_{1+e})/\{\pi_{+1e}\pi_{+0e}\pi_{1+e}\pi_{0+e}\}^{-1/2}\, .\label{Defccor} \end{equation}
The following adapts  a  previous result to be used  below for palindromic Ising models.

\begin{thm} \label{ConParCor}{\rm Baba, Shibata and Sibuya (2004).}  The partial correlation $\rho_{12.N\setminus\{m,1,2\}}\!$ 
 coincides with the conditional correlation $\rho_{ij|e} $  if  and only if   $\E(A_i,A_j|A_e)$  is a  linear function of  the vector variable $A_e$ and  $\rho_{ij|e} $   is constant at all level combinations $e$. \end{thm}
Theorem 2 and   Proposition \ref{CEqLogLin} imply  for the equivalence of  $\rho_{ij|e} $ and $\rho_{12.N\setminus\{m,1,2\}}$, that no  node chosen as a response  should have more than two neighbours.
This is reached  with any proper node-set elimination scheme for  hollow trees and, as we shall see, only for these palindromic Ising models.

\subsubsection*{\small 3.3 Palindromic Ising models generated over chordless cycles} 

We start this section with  a 4-cycle having its  nodes  labelled as in the left graph of Figure \ref{Cycles4&5}.
  \begin{figure}[H]
 \centering 
\includegraphics[scale=0.42]{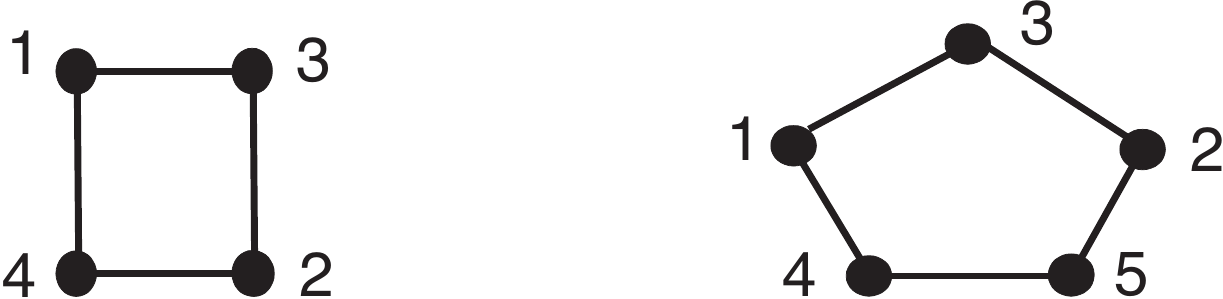}
{ \caption{ \small{\it left: labeled 4-cycle; right:   labeled 5-cycle} }\label{Cycles4&5}}
  \end{figure}  
\noindent  Thus, the missing edges  are for  pairs (1,2), (3,4). An example  is in the next table.

$$\centering
{\small
\left[\begin{array}{r r r r r r r r r r  r r r r r r r r r } 
  ijkl \text{ with }l=0:&0000 & 1000 & 0100 & 1100 &0010& 1010 & 0110 & 1110  \\
556\, \pi_{ijkl}\,: &180 & 36& 5& 1& 20& 15& 12& 9\\
\end{array}\right]	}
$$

The following matrix on the left contains $\lambda_{ij}$ in the lower and  $\tau_{ij}$ in the upper-triangular part, the matrix on the right has $\rho_{ij}$ in the lower and $\rho_{ij.k\ell}$ in the upper-triangular part. 

{\small
$$    
\begin{array}{cccc}\hline
                           0 & 0 &  0.319&0.442 \\
                           0&0& 0.646&0.771\\
                           0.330 &0.768& 0& 0 \\ 
                            0.474 &  1.024  & 0& 0\\  \hline
                            \end{array} \quad  \fourl
                            \begin{array}{cccc} \hline
                            1 &  0 &  0.226& 0.286 \\
                           0.511 & 1 & 0.459&0.648\\
                     0.504   &  0.705  &   1& 0\\ 
                     0.561 & 0.806&   0.597& 1\\
                     \hline
                         \end{array} \quad
                       $$        
}   
The two conditional independences $1\ci 2|\{3,4\}$ and $3\ci 4|\{1,2\}$ show as zeros  in usual measures of dependence  for 
the contingency table, $\lambda_{ij}, \tau_{ij}$,  and   in the partial correlations, $\rho_{ij.k\ell}$. 

For the chordless cycle here, all elements of  $\Sigma^{-1}$ relate  to  joint linear  regressions of two uncoupled nodes  on their neighbours.  For  $a=\{1,2\}$ as in Figure \ref{Cycles4&5}, the neigbours of nodes $1,2$ are  both in $b=\{3,4\}$. With $N=(a,b)$ and $m=\emptyset$,  equation   \eqref{PcorRegcCon} gives for row 1 of $\Sigma^{-1}$
\begin{equation}  \sigma^{11}=1/\sigma_{11|b}, \nn (\sigma^{13}\!, \;\sigma^{14})=-\sigma^{11}\Pi_{1|b}, \label{Conc&LReg}\end{equation}
where  $\Pi_{1|b}=(\beta_{1|3.4},\; \beta_{1|4.3}),\nn \sigma_{11|b}=1-\Pi_{1|b}\Sigma_{b1}$ and  for row 2, the analogous results.
 It is shown in App. B that linear regression coefficients are positive multiples of the $\tau$s in an almost unchanged form  
of equation \eqref{BetaTauLam} which holds for trivariate distributions. Similarly,  partial correlations  are then also positive multiples of  the $\tau$s.

The following equations, \eqref{InducedCorr4Cyc} and \eqref{InducedCorr5Cyc}, 
give the induced marginal correlations for pair $(1,2)$ in both cycles of Figure \ref{Cycles4&5} in terms of the hypetan interactions,  $\tau_{ij}$. A formal proof for
the 4-cycle is in App. A, while the following  explanations  are  to provide  some more intuitive insights.
 
In the  labelled 4-cycle, we have   $\tau_{13}=\tau_{13\backslash 4}$ and  $\tau_{23}=\tau_{23\backslash 4}$ due to $3\ci 4|\{1,2\}$.  Marginalising over node 4  in  the subgraph  $1\ful 4 \ful 2$, which is a so-called   transmitting {\sf V}, induces $\tau_{12\backslash 4}=\tau_{14}\tau_{24}$.
By applying next equation \ref{rectau} to the marginal table of $A_1,A_2,A_3$, one gets
for the margin  of $A_1, A_2$:
\begin{equation}
\rho_{12}^{*}= \tau_{12\backslash 34}=(\tau_{14}\tau_{24}+\tau_{13}\tau_{23})/\text{const., \n  with  const.}=1+  \tau_{13}\tau_{14}\tau_{23}\tau_{24}. \label{InducedCorr4Cyc} \end{equation}
The linear 4-factor interaction is  instead a multiple of the sum of   products  of $\tau_{ij}$s for disjoint edges:
$$ \xi_{1234} =  (\tau_{13}\tau_{24} + \tau_{14}\tau_{23})/  \text{const., \n  with  const.}=1+  \tau_{13}\tau_{14}\tau_{23}\tau_{24}.$$
However, this interaction is not relevant  if the undirected graph is oriented into a Markov-equivalent LWF
chain graph, that is, for instance,  by letting an arrow point from node 4 to 1 and from node 2 to 3. The distribution  generated over this graph has factorisation $f_{13|24}f_{24}$ and  satisfies for response nodes 1 and 3  the two independences $1\ci 2|34$ and $3\ci4|12$.

 For the labelled 5-cycle of Figure \ref{Cycles4&5}, a minor modification of the 4-cycle  argument  gives:
\begin{equation} \rho_{12}^{*}= \tau_{12\backslash 345}=(\tau_{14}\tau_{45}\tau_{25}+\tau_{13}\tau_{23})/\text{const., \n  with  const.}=1+  \tau_{14}\tau_{45}\tau_{25}\tau_{13}\tau_{23}.  \label{InducedCorr5Cyc} \end{equation}
One knows  that $\tau_{j3}=\tau_{j3\backslash 45}$  for $j=1,2$ since $\{4,5\}\ci  3|\{1,2\}$
 and  $\tau_{12\backslash 45}=\tau_{14}\tau_{45}\tau_{25}$ by closing the transmitting {\sf V}s along the path $1\ful 4\ful 5\ful2$.  Next, one  applies   equation \eqref{rectau},   again as above, to the table of  $A_1,A_2,A_3$, which remains after marginalising over $4,5$. 
 
 These  arguments for the 4- and 5-cycle lead  directly to the following result for all chordless cycles.  \begin{prop} In a chordless $h$-cycle of a palindromic Ising model, the induced $\rho_{k\ell}^{*}$ for an uncoupled node pair $(k,l)$ results by tracing the two paths, connecting $k,\ell$ in the cycle,  in terms of  products of $\tau$'s, attached to the $h$ edges present in the cycle. \label{InducingPaths}
\end{prop}

A matrix form for induced correlations uses  instead partial correlations. For this, 
we denote  the  overall matrix of partial correlations  by $\theta$. The relation between $\Sigma$ and $\theta$ becomes transparent by using  the identity matrix $I$ and standardising, that is changing
the elements $m_{ij}$ of a symmetric matrix $M$ to give with the elements of $M_{\rm std}$ as $m_{ij}/\sqrt{m_{ii}m_{jj}}$:
\begin{equation}  \Sigma^{*}=K_{\rm std}  \text{ with }    K= (2 I -\theta)^{-1}.\label{IndCov} \end{equation}
This holds in unchanged form for all larger chordless cycles. 

With $|M|$ denoting the determinant of $M$, the product   $\Sigma^{-1}|\Sigma|$ gives two equations in terms of marginal correlations which may be solved numerically for a set of given data to obtain $x=\rho_{12}^{*}$ and  $y=\rho_{34}^{*}$:
\begin{equation}  x(y^2-1)-c_1 y+ c_2 =0, \nn \n y(x^2-1)-c_1 x+ c_3 =0 , \label{polyEq}\end{equation}
where $c_1=\rho_{13}\rho_{24}+ \rho_{14}\rho_{23}$ is the sum of products of disjoint edges,   $c_2=\rho_{13}\rho_{23}+ \rho_{14}\rho_{24}$ relates to the two paths connecting pair $(1,2)$
and $c_3$ relates,  analogously,
 to the two paths connecting $(3,4)$. The five possible solutions reduce to one if one constrains them to be within a range, say of $\pm 0.2$, of the given $\rho_{12}$ and  $\rho_{34}$. For the example to a diamond graph in the previous section, the solution changes $\rho_{12}=0.523$ to $\rho_{12}^{*}=0.545$ and  $\rho_{34}=0.669$  to $\rho_{34}^{*}=0.590$.

To generate the probabilities to the 4-cycle in Figure \ref{Cycles4&5}, there are then two main options. One may  (1)  fit  the four sufficient two-way tables with some standard iterative procedure or  (2)  generate first the probabilities in the marginal tables for nodes $1,3,4$ and $2,3,4$  using  the solutions  of equations  \eqref{polyEq}  and  exploit next that the  independence constraint for   pair $(1,2)$  gives  $ f_{1234}= f_{1|34} f_{234}.$

The two trivariate tables used in this expression may  be obtained in  terms of simple correlations  and equation \eqref{rhopi},  in  terms of  partial correlations and equation \eqref{ParcortoProb} or  simply as  marginal three-dimensional tables when  an observed table with skewed margins has been transformed into a  saturated 
palindromic Ising model, as described in section 3.5. There are analogous, albeit a bit more complex extensions  to larger chordless cycles, not given here.

\subsubsection*{\small 3.4 Some properties of Ising models with hollow tree-structure} 

We recall from  Section 1, how parameters of  Bernoulli distributions are defined in equation \eqref{DefInt} and that in palindromic Ising models, the 1-factor
terms are zero, in addition  to all higher than 2-factor log-linear interactions. Equations \eqref{CondInd}  define independences and equations  \eqref{DefDep} point to the dependences in Ising models with an edge-minimal graph structure.
 
From Section 2, we remember that models with hollow-tree structure have a unique set of prime graphs which consists of   edges or cycles, attached to each other at cut-sets which are single nodes  or edges. The unique  set of cut-sets  arises with any  proper node-set elimination scheme giving sequences of outer node sets as responses residing in  a single prime graph.  

Binary thin trees as well as  bulged  trees are a subclass of the hollow trees but  contain exclusively chordal-graph models.
For  thin trees with non-symmetric margins,  there is an extensive literature in graph theory, in phylogenetics and in machine learning. Bulged  trees have been studied under the name  `dino graphs' by Loh and Wainright (2013) for general multinomial and for Bernoulli distributions.

In general Bernoulli distributions with a thin or extended tree-structure, one may, by   Proposition \ref{MargPrimeG}, reduce the concentration graph
to  subgraphs, such as  to the tree trunks in Figure \ref{TreeTrunks}, which have a subset  of nodes $a\cup b\cup c$ satisfying the  conditional independence constraint  $a\ci b|c$. The nodes in $a$ and in $b$ are then the  outer nodes of two prime graphs, set $c$  separates $a$ and $b$;  the smallest remaining separator  is either a cut-set or a prime graph. 
Then, the  joint probability distribution in this  marginal distribution,  $f_{a\cup b\cup c}$,  factorises as   
\begin{equation} f_{a\cup b\cup c}=f_{a\cup c}f_{b\cup c}/f_{c}=f_{a|c}f_{b|c}f_c\,.  \label{FactMarg}\end{equation}
 For  our extended trees,  several   special features, based on linear relations, are to be summarised next. In particular, we shall see that all conditional independence constraints -- and only these --  are reflected  for the palindromic hollow trees  in zero elements, $\sigma^{ij}=0$,  of $\Sigma^{-1}$. Each $ij$-edge present in the graph shows instead as $\sigma^{ij}\neq 0$.

 For palindromic Ising models with hollow tree-structure and factorisations as in equation \eqref{FactMarg}, there are  block-triangular matrix  decompositions of $\Sigma^{-1}$ for  $N'=a\cup b\cup c$
 just  as for joint Gaussian distributions;  see Wermuth (1992, equation (3.5)):
\begin{equation} \Sigma^{-1}=
\varphi_{a}(\Sigma^{aa})^{-1}\varphi_{a}^{T} +
\varphi_{b}(\Sigma^{bb.a})^{-1}\varphi_{b}^{T} +
\varphi_{c}(\Sigma^{cc.ab})^{-1}\varphi_{c}^{T}  \label{BlockDiagComp}
\end{equation} with  \[ \varphi_{a}=\left( \begin{array}{l}
\Sigma^{aa}\\
\Sigma^{ba}\\
\Sigma^{ca}
\end{array}\right)\!,\n
\varphi_{b}=\left( \begin{array}{l}
0_{ab}\\[1mm]
\Sigma^{bb.a}\\
\Sigma^{cb.a}
\end{array}\right)\!, \n
\varphi_{c}=\left( \begin{array}{l}
0_{ac}\\
0_{bc}\\[1mm]
\Sigma^{cc.ba}
\end{array}\right).\]
where for instance, $\Sigma^{ba}$ is the submatrix for $(b,a)$ in a concentration matrix for $a\cup b \cup c$, $\Sigma^{cb.a}$ is the submatrix for $(c,b)$ in a concentration matrix obtained 
after marginalizing over $a$ and $\Sigma^{cc.ba}$ is a concentration matrix for $c$ after marginalizing over $a$ and $b$. 

When the nodes in $a$ and in $b$ are the outer nodes of two prime graphs of a hollow tree  and $c$  is a prime graph or a cut-set separating $a$ and $b$, then $a\ci b|c$, 
 is an `independence constraint between prime graphs' and  shows  as $\Sigma^{ab}=0$, while, for instance,  when  $a$, $b$ and $c$ have  more than three nodes, there are cycles  within $a$, $b$ and $c$ in the  graph,  which represent  `independence constraints within prime graphs'  and  show as cycles  within $\Sigma^{aa}$, $\Sigma^{bb.a}$ and $\Sigma^{cc.ab}$.  One consequence of this representation  is the following result proven in App. C.
\begin{thm}   \label{IsGaussConc} A quadratic exponential distribution for symmetric binary variables is generated over an edge-minimal hollow tree if and only if   (1) the
Markov structure  of its concentration  graph is defined by  the set of zeros in its   overall partial correlation matrix  and (2)  within each of its  prime graphs, all conditional  correlations agree  with  the partial correlations. \end{thm}

Theorem \ref{IsGaussConc} states essentially that such joint distributions are determined by linear relations. 
This  implies  that dependences along paths can be combined by marginalising since linear relations are traceable; for this notion see here  App. D or  Wermuth (2012). For instance in hollow trees with chordless cycles as prime graphs, the  simple correlations induced for the missing edges within a cycle may be obtained in terms
of partial correlations or, in a more compact way,  in terms of the hypetan interactions, as in equations  \eqref{InducedCorr4Cyc} and \eqref{InducedCorr5Cyc}; see also Proposition \ref{InducingPaths}.  For a 4-cycle, the linear relations lead to
 equations  \eqref{polyEq}, expressed  in  terms of the correlations for
 the edges present in the graph and give  solutions for the induced correlations.

Theorem \ref{IsGaussConc} is  important for judging a hypothesized hollow-tree structure in a new sample
for two main reasons.  First, one can recog\-nize the pairwise independences defining the Markov structure in the overall partial correlation matrix, $(2I-\Sigma^{-1})_{\rm std}$, which is easily estimated for reasonably large sample sizes and contains elements ranging at most within $\pm 1$.   Second, one can judge the strength of  each pairwise dependence by the  partial correlation computed from small submatrices of the estimated  
 covariance matrix  corresponding the  unique subset of prime graphs for the hollow tree. 
 
 Together with the equivalence in Proposition 
\ref{CEqLogLin} for main-effect logit and main-effect linear regressions, Theorem 3 also leads to the following results.

 \begin{coro} \label{Cor1}
   The joint distribution of a palindromic Ising model with   thin- or bulged-tree structure
 can  be generated with  linear regressions by using any  single-node elimination scheme.  \end{coro}

 \begin{coro}  \label{Cor2}
Palindromic Ising models with hollow-tree structure are traceable, but may contain
 cycles both  in partial  correlations  and  in marginal correlations. 
\end{coro}
 \begin{figure}[H]
 \centering
\includegraphics[scale=.34]{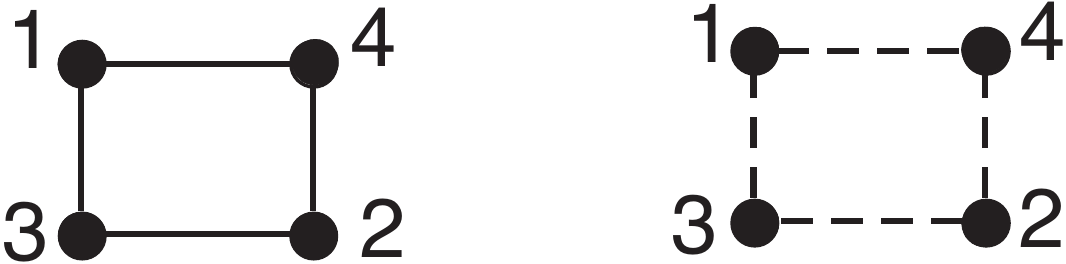}
 \caption{ \small{\it  left:  4-cycle in concentrations; right: 4-cycle in covariances; $i \dal j$ indicates $i\dep j$}} \label{Cycle4ConCov}
  \end{figure}
The more restrictive requirement of a distribution to have no additional independences than those indicated by the graph, is not 
satisfied by general Gaussian distributions; see for an example Wermuth (2012, Section 2.5). Both cycles of Figure \ref{Cycle4ConCov} hold also in the next table:

$$
\centering
{\small
\left[\begin{array}{r r r r r r r r r r  r r r r r r r r r } 
  ijk\ell \text{ with }\ell=0:&0000 & 1000 & 0100 & 1100 &0010& 1010 & 0110 & 1110  \\
80\, \pi_{ijk\ell}\,: &9 & 9& 1 & 1& 1& 9& 1 & 9\\
\end{array}\right]  }
$$
 The following matrix on the left contains $\lambda_{ij}$ in the lower and  $\tau_{ij}$ in the upper-triangular part, the matrix on the right has $\rho_{ij}$ in the lower and $\rho_{ij.k\ell}$ in the upper-triangular part. 

{\small
$$    
\begin{array}{cccc}\hline
                           0 & 0 &  0.635&-0.635 \\
                           0&0& 0.635&\n \;0.635\\
                           \; \n 0.749 &0.749& 0& 0 \\ 
                            -0.749 & 0.749  & 0& 0\\  \hline
                            \end{array} \quad  \fourl
                            \begin{array}{cccc} \hline
                            1 &  0 &  0.452& -0.452 \\
                           0 & 1 & 0.452&\n\; 0.452\\
                     \n \; 0.452   &  0.452  &   1& 0\\ 
                    - 0.452 & 0.452&   0& 1\\
                     \hline
                         \end{array} \quad
                       $$        
}  
Notice from equation \eqref{InducedCorr4Cyc} that a special parametric constellation induces here the zero marginal correlations $\rho_{12}^{*}=0$ and $\rho_{34}^{*}=0$. Such constellations have been named parametric cancellations, when they were discussed for Gaussian distributions by Wermuth  and Cox (1998). It is an open question  whether such constellations may  characterise 
 processes of special interest. Here, the marginal and the partial correlations coincide, in addition, and there is no 4-factor linear interaction. More generally, one knows for instance  the following.
\begin{prop} If a palindromic hollow-tree Ising model has overall exclusively positive or zero partial correlations, then  there are no path cancellations
and no correlation is negative.
\end{prop}

\begin{proof} The claim follows  with Theorem \ref{IsGaussConc}, Proposition \ref{InducingPaths} and results for the total positivity of joint distributions; see Fallat et al. (2017), Karlin and Rinott (1983), Bolviken (1982). \end{proof}
In Proposition 5, `correlation' means all possible types of correlations, partial as well a simple ones. 
 In such  distributions, there will  in particular  never be any effect reversal, i.e. no Yule-Simpson paradox, and for our connected graphs,  all simple correlations are  positive.

\subsubsection*{\small 3.5 Transformations between general and palindromic Ising models}
 For  binary two-way tables with probabilities $\pi\T=[\alpha, \beta, \gamma, \delta]$, there is a unique way to move from the skewed margins to symmetric ones, preserving the odds-ratio; see Cox (2006, Sec. 6.4), Palmgren (1989). The next table divided  by $2(\sqrt{\alpha \delta}+\sqrt{\beta \gamma}) $ gives the transformed probabilities
 \begin{equation}
{\small
\left[\begin{array}{r r  }   
\sqrt{\alpha \delta}& \sqrt{\beta \gamma}\\
\sqrt{\beta \gamma}& \sqrt{\alpha \delta}\\
\end{array}\right]}\,.	\label{OdrtoPi}
\end{equation} 
 This builds on the result by Edwards (1962) that the odds-ratio and only functions of  this canonical parameter vary
 independently of moment parameters, here of the marginal frequencies.  The transformation may equivalently be obtained 
in terms of log-linear parameters: $(\lambda_{\emptyset}, \lambda_1, \lambda_2, \lambda_{12}) \to (\lambda_{\emptyset}^{'},0,0,\lambda_{12})$,
where one recalls that for $-1,1$ coding  $\lambda_{12}=\fracsfourth \log(\alpha \delta)/(\beta \gamma)$.

The asymptotic variance, $\mathrm{avar}$, of  $\hat \lambda_{12}$ is best found with a Poisson-based formulation of the unconstrained maximum likelihood equations. One uses that for an estimate $\Sigma_k c_k \log(n_k/n)$ from $n$ observations with cell frequencies $n_k$,  the asymptotic variance, derived  with the delta method as $\Sigma_k c_k^2/\E (n_k)$;  see Bishop, Fienberg and Holland (1975, p. 495). It turns out to be:  
\begin{equation}
n \;\mathrm{avar}(\hat \lambda_{12}) = 4+4[\cosh(2\lambda_{12}) \, \{ \cosh(2\lambda_{1}) +\cosh(2\lambda_{2})\} + \cosh(2\lambda_{1})  \cosh(2\lambda_{2})]\,.
  \label{eq:one}
\end{equation}
 By setting $\lambda_1= \lambda_2 = 0$ in \eqref{eq:one},  the palindromic case is obtained, where the correlation coefficient $\rho_{12}$ is a 1-1 function of $\lambda_{12}$.  In the symmetric case,  the variance is 16   under independence and typically not much larger, otherwise, while it can  increase considerably for
 skewed margins. 
  If a hollow-tree structure is well-fitting to a table which preserves all  observed marginal odds-ratios, it might therefore happen that an edge is no longer significant 
 for the original data. But, by removing any edge of a hollow tree the resulting structure is still a hollow tree.

 We apply the  transformation of equation \eqref{OdrtoPi} to all $2 \times 2$  tables and  use the additional knowledge that for any hollow-tree Ising model,  the two-way tables of edges present in the graph,
form the  set of minimal sufficient  statistics. 
But for any given set of data,
one  needs to check first whether it is close   to an  Ising model  without any independence 
 constraints. For this, all  higher than 2-factor log-linear interactions have to be so small 
 that they can be ignored.  
 
 With the transformed two-way tables,  a  palindromic Ising model  without any independence constraints results  with any
 iterative proportional  fitting  procedure.    
 Next, one   exploits  the computationally attractive features of this distribution to  find a possibly well-fitting  hollow-tree structure. 
Given the selected subset of edges present, the set of minimal sufficient statisticis for  a hollow-tree structure in the original data  is  the same 
subset of the observed $2\times 2$ tables.

It cannot be extended to  more than two variables to  symmetrise an observed  table as above,  such that the resulting  palindromic  Ising model preserves  the marginal odds-ratios, 
or by setting  to zero  the main effects  in a multivariate regression for categorical responses, due to  Glonek (1986).
 One could use instead a multivariate logistic parametrization; see Glonek and McCullagh (1985),  Bergsma and Rudas (2002), Marchetti and Wermuth (2017), but this would lead to considerably more demanding computations.

\subsubsection*{\small 3.6 Model selection for  Ising models with hollow tree-structure} 
 As is known, for  general  Ising models,  the set of minimal sufficient statistics  consists of  observed two-way tables for all edges present in the graph; see Birch (1963, equations (3.5) to (3.6)),  while for  palindromic Ising models, it contains  the corresponding symmetrised 
two-way tables. The  latter  are  averages of the observed counts for $p(\omega)$ and $p(-\omega)$; see Marchetti and Wermuth (2016, equation 4.6) or they are a subset of the two-way tables preserving the marginal odds-ratios of an observed  general Ising model. 

 The mle-computation  for a general, unconstrained  trivariate Ising model  requires already an ite\-rative procedure.
Therefore,  computational complexity will  increase for  Ising models of bulged-tree structure with 
 many triangles as prime graphs.  By contrast, mle-fitting of a palindromic Ising model is in this class
  in closed form, provided only one uses a proper single-node elimination scheme; see Corollary  
\ref{Cor1}.

 In
a palindromic Ising model with a  fattened-tree  structure, one gets with equations  \eqref{FactMarg} and \eqref{BetaTauLam}, the likelihood-ratio statistic for  $a\ci b|N\setminus \{a\cup b\}$ reduced to one for $a\ci b|c$:
\begin{equation} \chi^2=2(\log \, n_{a\cup b\cup c} - \log \, n_{a\cup c} - \log\,  n_{b \cup c}+ \log \, n_c);  \label{ChiSquare}\end{equation}
see  the discussion of hypothesis $H_{d3}$ by Birch (1963). 
Thus, this type of goodness-of-fit test,  a  `test between prime graphs', can be carried out in smaller tables and without computing the mle. 
This is an attractive property in general, but it is especially important for  very large data sets, where it avoids working with  extremely sparse tables.

Equation  \eqref{ChiSquare}  becomes even more useful  for  Ising models with a hollow-tree structure. 
\begin{prop} For an  Ising model with hollow-tree structure, goodness-of fit tests are of two types:  for conditional independences between  and within prime graphs.
\end{prop}
The `tests between prime graphs' reduce  for instance to closed-form tests  between two covering, completed prime graphs,  attached to each other at an edge or a single node. 
 A `test within prime graphs'  is for a saturated Ising  model in a triangle  or for the fit of 
a chordless $h$-cycle for $h>3$. Such smaller tables make estimation and testing much more reliable; see  Altham (1984).

For data which are close to a palindromic Ising model, any standard iterative method estimates the 2-factor  log-linear interactions even though it is a so-called non-hierarchi\-cal model.  The 
mle of the probabilities relates  to the mles of the log-linear interactions and 
 the linear interactions  via the same relation in  equation \eqref{DefInt}, that hold for the parameters. This is an extremely attractive property of maximum-likelihood estimates. 

There are further options when a hollow-tree structure  is given as a hypothesis. Then, the 
estimates  between and within prime graphs can be obtained with  linear regressions on at most two regressors.
One may alternatively first get  the mle of  $\Sigma^{-1}$ with any  Gaussian estimation routine; see  Speed and Kiiveri (1986), Sadeghi
and Marchetti (2012), Lauritzen, Uhler and Zwiernik (2017).  Given the selected subset of variable pairs which generate 
a joint Gaussian distribution with hollow-tree structure,  one can  estimates the same structure 
for Ising models by using  the corresponding  subset of two-way tables.

However, for sample sizes larger than the number of variables, the latter  is not an option.  In these cases, one might first  try to learn a concentration graph structure as for Gaussian distributions but using the marginal correlations in all two-way tables, transformed to symmetry by preserving  the observed odds-ratios as described in the previous section.

One robust search technique is due to Castelo and Roverato (2006) which assumes however that all independences of the graph and only these hold in the generated distribution. If the resulting graph is a hollow tree, then the block-triangular factorisation of
the concentration matrix can  again be exploited to estimate  probabilities in the palindromic Ising model and in a starting general Ising model; see also  the  data example in the next section.
 For very large data sets this is especially important since it reduces the  fitting  to small subgraphs.

\subsubsection*{4 An Ising model with a hollow-tree, for symmetric and for skewed margins} 

We use data from a larger prospective study  in Germany, designed to develop strategies for counselling prospective students;
see Weck (1991). 
We apply first a standard graphical check,  due to Cox and Wermuth (1994b), to reassure us that the data are close to an Ising model. \begin{figure}[H]
\centering
\mbox{\hspace{-1cm}\includegraphics[scale=.24]{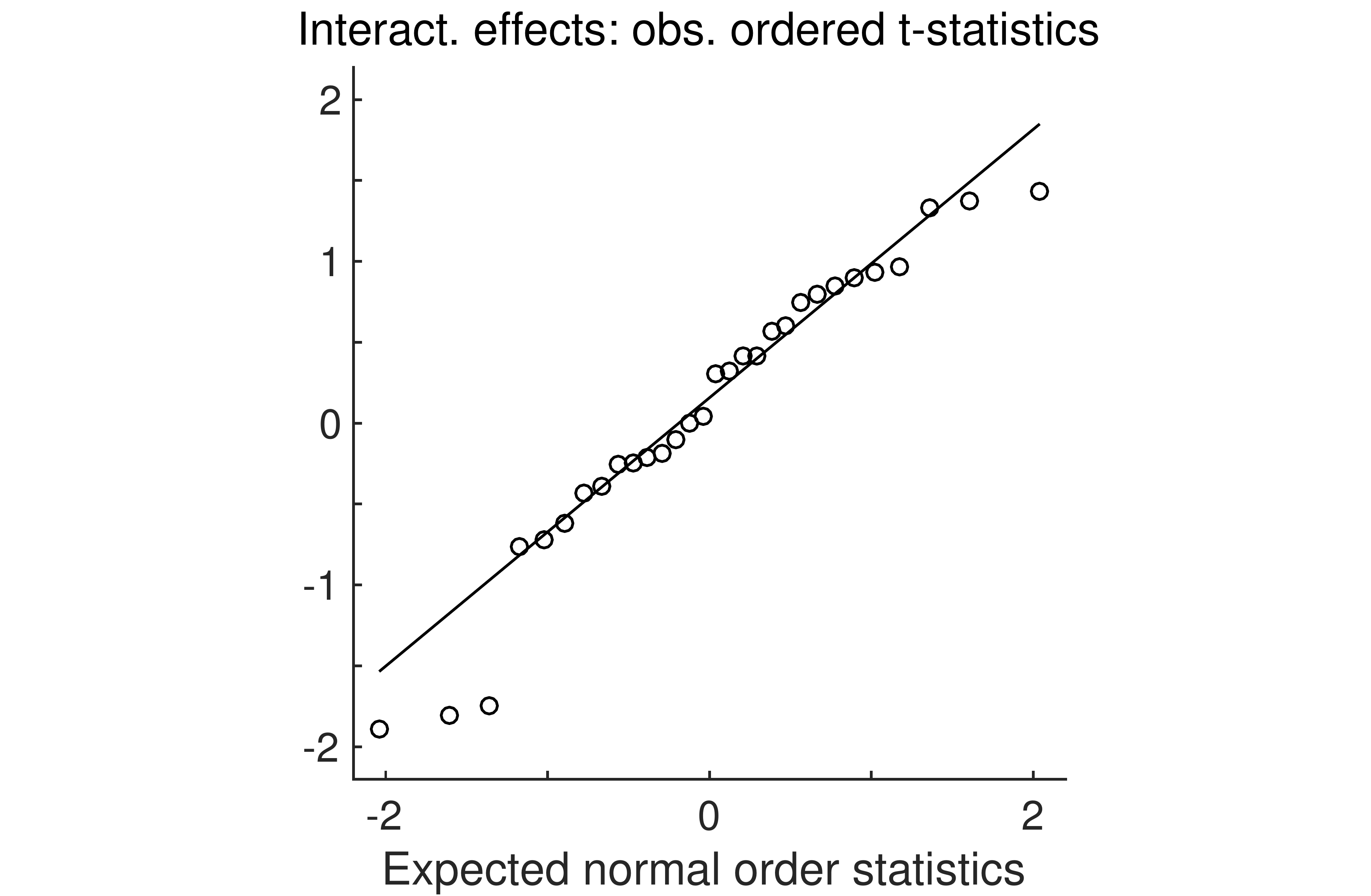}\hspace{-2cm} \includegraphics[scale=.24]{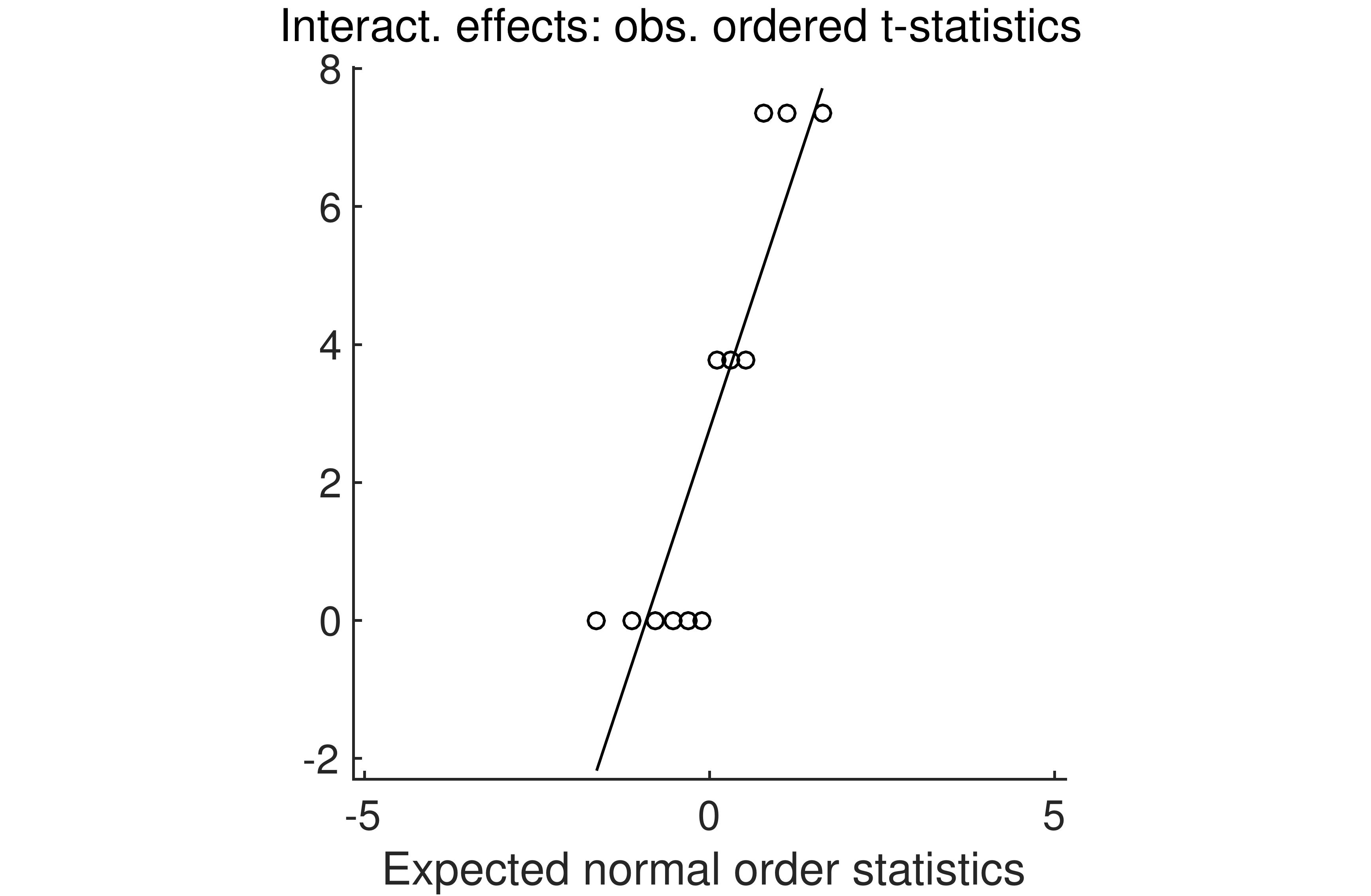}}
 \caption{ \small{\it  left:  achievement data; right: strong 3-factor interactions in the first table above}} \label{InteractEffects}
  \end{figure}

The left plot of Figure \ref{InteractEffects} is for this  set of data having six binary variables.
  The approximate $t$-values for linear 3-factor interactions lie well within the range of (-2, 2), hence do not indicate any strong 3-factor interaction and we proceed with an analysis. By contrast, the corresponding plot for the  counts given in the first table of this paper indicate  strong interactions. There, the plot shows strong ones to be within subset  (1,2,3) and the weaker ones within  (1,2,4).

The five variables are the average grade reached in years 11 to 13 at high school  in median-dichotomized 
form ($A_1$, level 1 means achievement above median: 50\%).
Further  variables  are $A_2$:= poor integration into high-school classes (yes: 10.3\%), $A_3$:= high-school class repeated (yes: 34.2\%), $A_4$:= change of primary school (yes: 20.1\%) and $A_5$:= education of father (at least `Abitur', i.e. high-school completed: 42.7\%).
Variables $A_2$ to $A_5$ are potential regressors for $A_1$. By  the time-order, variables $A_4, A_5$ are further in the past than  $A_2, A_3$. The observed counts are in the next table.
$$\setlength{\arraycolsep}{.9\arraycolsep}
\centering
{\small
\left[\begin{array}{r r r r r r r r r r  r r r r r r r r r }   
n_{ijk\ell m} \text{ with }m=0:&  275& 419& 30& 22& 238& 99& 37& 13& 46& 67& 12& 3& 49& 21& 7& 2\\
n_{ijk\ell m} \text{ with }m=1:& 161& 292& 19& 29& 142& 56& 29& 7& 49& 102& 4& 9& 59& 23& 12& 6\\
\end{array}\right]	}
$$

After  applying equation \eqref{OdrtoPi} to all two-way tables, an iterative proportional fitting algorithm gives  probabilities  of a saturated  palindromic Ising model. For  the above data,
 the following matrices contain  $\tilde{\rho}_{ij}$ in the lower and $\tilde{\rho}_{ij.k\ell m}$ in the upper-triangular part; on the left without,
  on the right with independence constraints.

{\small
$$    
\begin{array}{ccccc}\hline
                           1 & -0.098 & -0.318&0.007&\n \;0.044 \\
                           -0.138&1& \n \; 0.108&0.029 &\n \;0.058\\
                           -0.333  &\n\;0.148& 1& 0.050& -0.016\\ 
                            -0.007 &\n\; 0.045  & \n \;0.052& 1&\n\; 0.166\\ 
                          \n \; 0.042&\n \;0.056&-0.015& 0.167&1 \\
                             \hline
                            \end{array} \quad  \fourl
 \begin{array}{ccccc}\hline
                           1 & -0.096 & -0.319&0& 0 \\
                           -0.138&1& \n \; 0.108&0 &\n \;0.054\\
                      -0.333 &\n\;0.148& 1& 0.047& 0\\ 
                                -0.018  &\n\; 0.017  & \n \;0.052& 1&\n\; 0.167\\ 
                             - 0.010&\n \;0.056&\n\;0.017& 0.167&1 \\
                             \hline
                            \end{array} \quad  
                       $$      }  
All partial correlations  which
are smaller than $|0.045|$ in the left table have been set to zero   in the right table: these are for pairs (1,4), (1,5), (2,4), (3,5).
 The induced marginal correlations for these pairs differ from the unconstrained marginal correlations, while for all other pairs, the 
 marginal correlations in the saturated model on the left match the marginal correlations on the right. This coincides
 with the defining property of a concentration graph model for Gaussian distributions; see Dempster (1972).
 
 Figure \ref{DropAchievm} shows the well-fitting hollow tree on the left, an equivalent oriented version on the right.
   This hollow-tree structure fits also the original data where
 the set of  minimal sufficient statistics are  the observed two-way tables of the edges present in the graph. A split of the test into two
 parts, fitting first an Ising model without independence constraints gives  a likelihood-ratio statistic of $\chi^2=17.65$ on 16 degrees of freedom (df)
 and  second, the additional independences contribute $\chi^2=5.78 $ on 4 df.
  \begin{figure}[H]
 \centering 
\includegraphics[scale=0.45]{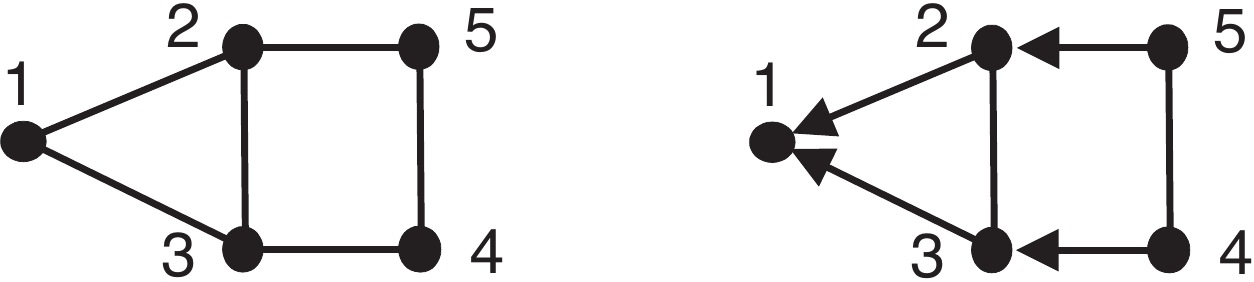}
{ \caption{ \small{\it left:  hollow tree, well-fitting for achievement data; right:  equivalent oriented graph} }\label{DropAchievm}}
  \end{figure}

The same conclusion of a good fit can be reached by testing the `hypothesis between prime graphs', here $1\ci \{4,5\}|\{2,3\}$,  using equation \eqref{ChiSquare} on the observed table and the two  hypotheses within prime graphs in  the  smaller tables of variables $1,2,3$ and $2, 3,4,5$. These give, respectively,
 $\chi^2=16.75 $ on 12 df, $\chi^2=1. 16$ on 1 df and $\chi^2=5.52$  on 7 df. Of course, these types of decomposition become especially useful  only in very large tables. 
 
 The next table shows of the  fitted  log-linear terms, all that are nonzero  for the general Ising model.  For the latter it also gives corresponding  t-statistics.
\begin{center}
\setlength{\tabcolsep}{3.2pt}
\small
\begin{tabular}{lrrrrrrrrrrrrrrrr} 
\hline
Int. type &$\emptyset$& 1& 2& 3& 4& 5& 12& 13& 23& 25& 34& 45\\
pal. Ising& 4.204 &        0   &  0&    0   &0 &   0  & -0.102&   -0.335&   0.116&    0.055&   0.051&    0.169\\
gen. Ising&  3.458 &   0.198 & 1.064&  0.245 &   0.666  &  0.003 &  -0.104  & -0.342  &  0.116  &  0.056 &   0.051   &  0.169\\
t-statistics&  -&  5.27 &  29.88& 6.16  &   24.80 &    0.00     &  -2.81  &  -14.45  &  3.20  &  2.34 &  1.92  & 7.09\\ \hline
\end{tabular}    
\end{center}
The log-linear 2-factor terms agree closely but not fully  for the palindromic and the general Ising model.  To interpret  the  generally weak dependences,
we use some of the time order.

High achievement  is less likely when a high-school class had been repeated and when a student integrated poorly
into high-school classes. To repeat a high-school class  is more likely with  such poor integration 
and  with a previous change of primary school. The  risk for poor integration increases in the case of better educated fathers
and with  repeating classes. Finally, the  risk  to change primary school is higher for better educated  fathers.

\subsubsection*{5 Discussion}
We have studied Ising models, which have simple, undirected and connected probabilistic graphs for a finite 
set of nodes.   Each  edge missing   in these graphs   represents a   pairwise independence 
given all remaining variables and we let, in addition, each edge present indicate a corresponding pairwise dependence. Then,   the removal of any edge  from  a given graph leads  to a different model since  an additional independence constraint is introduced.

 Ising models form the subclass of lattice models, also called quadratic exponential  distributions,   
that is defined for  binary variables having levels $-1,1$.  For unconstrained marginal distributions of the binary variables, we speak of  general Ising models and when all  margins are symmetric of  palindromic Ising models.   The latter mimic  a central symmetry property of joint Gaussian distributions
and have the same probability for  level combinations $\omega$ and  $-\omega$.

In palindromic Ising models -- as in mean-centered Gaussian distributions --  there is only one parameter associated  with an edge present in the graph. In general, this parameter  is  a log-linear  interaction, not  based on  covariances.  
But, we have identified  a subclass of  graphs, named hollow trees,  in which palindromic Ising  models are generated  over oriented graphs such that all main-effect logit regressions can be   be replaced by linear main-effect regressions. 

To characterize  hollow trees, we  used early  graph theoretical results on the existence of  unique sets of prime graphs, their smallest complete separators, called cut-sets, and  corresponding  node-set  elimination schemes,
as well as Markov-equivalence conditions for different  types of chain graphs. The  only types of prime graphs of a hollow tree are chordless cycles, triangles or edges and its cut-sets contain one or two nodes. 
To recognize the tree-like  structure, one replaces nodes by prime graphs and edges by cut sets, then a `single path' connects each pair of `nodes'.

General Ising models with hollow-tree structure contain as subclasses, for instance, Markov chains and   tree-structures  used in phylogenetics, connected star graphs which can represent subscales in  item response studies and 
 non-decomposable models with chordless-cycle structures.   Due to their  computational complexity in mle-fitting, the last have not been used intensively in the past even though they may  capture  important spatial research questions. 
 
 This may start to change now when the attractive computational properties of palindromic Ising models with
 hollow-tree structure are exploited for any set of data  close to 
 a general   Ising model by applying  a one-time transformations  to a saturated palindromic Ising model which preserves all
 observed marginal odds-ratios. It may be necessary to check beforehand whether the data belong to a connected  graph, when one wants to study single  trees to understand a forest's build-up; for a history of  such  network-flow algorithms see Dinitz (2006).

Once a well-fitting hollow-tree structure has been obtained for a  palindromic Ising model, one knows
the subset of symmetric two-way marginal tables which generates this joint distribution. One
can then take the same subset of the observed two-way tables -- with possibly strongly skewed margins  --
to estimate the probabilities of a general Ising with the selected hollow-tree structure.
The  factorisation of the joint distribution with hollow-tree structure into subsets corresponding to  its prime graphs,
will permit one to compare results  for small subsets of the  variables with the available knowledge in the field about  these variables under study. This is always an  essential step towards progress
in substantive research, but especially so when data for  large numbers of variables are to be analyzed.

\subsubsection*{Appendix A:  Dependences induced by marginalising}
 {\bf \textit{\small Equivalent to the log-linear formulation of the palindromic Ising model}} in equation \eqref{DefPalIsLam} is 
\begin{equation} \Pr (\omega)= \text{const.} \, \txt{\prod}_{s<t} (1+ \tau_{st}\omega_s \omega_t),\nn \n 1/\text{const.}=2^d(1+  \txt{\prod}_{s<t}\tau_{st} ) \, ,  \nn -1<\tau_{st}<1 \, ,\label{DefPalIsTau}
\end{equation}
where  $\tau_{st}=\tanh ( \lambda_{st})$ are the hypetan interactions.
\begin{proof} The logarithm of equation (\ref{DefPalIsTau}) yields the sum $\sum_{s<t} \log(1 + \tau_{st} \omega_s\omega_t)$ where each summand can be expanded as 
\[
\begin{aligned}
\log(1 + \tau_{st}\omega_s\omega_t) &= \tau_{st}\omega_s \omega_t - (\tau_{st}\omega_s \omega_t)^2/2 + 
(\tau_{st}\omega_s \omega_t)^3/3 - \cdots\\
&=\omega_s\omega_t \left(\tau_{st} + \tau_{st}^3/3 + \tau_{st}^5/5 + \cdots \right) - \left(\tau^2_{st}/2 + \tau^4_{st}/4+ \tau^6_{st}/6 + \cdots\right)
\end{aligned}
\] 
and by  noting that  \(\omega^s = 1\) if \(s\) is even and \(\omega^s = \omega_s\)
if $s$ is odd, this  gives the explicit form
 \[
\log(1 + \tau_{st}\omega_s\omega_t) = \omega_s\omega_t \tanh^{-1}(\tau_{st}) + \log(1 - \tau_{st})/2 = \text{const.} + \omega_s\omega_t \lambda_{st},
\] 
so that  finally $\log \Pr(\omega) = \text{const.}' +  \sum_{s<t} \omega_s\omega_t \lambda_{st} .$ 
\end{proof}

We use  equation \eqref{DefPalIsTau} to derive several effects of marginalising in the following cycles:
\begin{center}
 \includegraphics[scale=0.45]{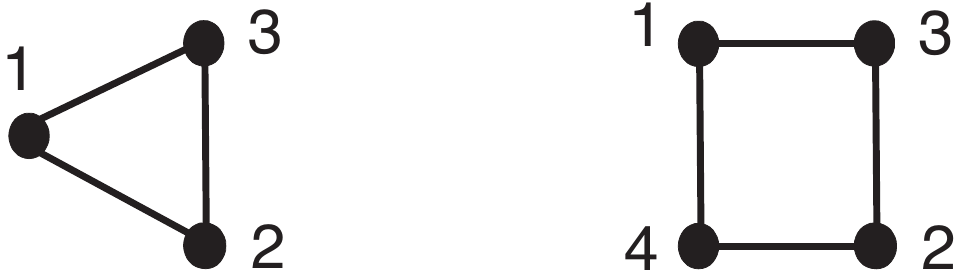} \end{center} 
  \noindent {\bf \textit{\small To obtain $\bm{\pi_{+jk}}$}} $= \txt{\sum}_i \pi_{ijk} $ and $\rho_{23}=\tau_{23 \setminus 1}$}}, one best uses   equation \eqref{DefPalIsTau}}
in the form:
\begin{equation}\pi_{ijk} = \text{const.}\,(1 + \tau_{12}ij)(1 + \tau_{13}ik)(1 + \tau_{23}jk), \nn {1/\text{const.}}=2^3  (1 + \tau_{12}\tau_{13}\tau_{23}) \label{Joint123}, \end{equation}
 to get the result via
\[
\txt{ \sum}_i (1+\tau_{12} ij)(1+\tau_{13} ik) = \txt{\sum}_i 1 + \tau_{12} j \txt{\sum}_i  \,i+ \tau_{13} k \txt{\sum}_i \,i    + \tau_{12}\tau_{13} jk \txt{\sum}_i  i^2
 = 2(1 + \tau_{12}\tau_{13} \,jk  ),
\]
\begin{equation}  \pi_{+jk}= \fracsfourth(1+  \rho_{23}\, jk),   \nn \rho_{23}=  \tau_{23 \setminus 1}=(\tau_{23} + \tau_{12} \tau_{13})/ (1 + \tau_{12}\tau_{13}\tau_{23}). \label{Marg23Tri}
\end{equation}
Similarly, when starting  with the chordless 4-cycle, one uses best equation \eqref{DefPalIsTau} in the form:
\[
\pi_{ijk\ell}\! =\text{const.}\,  (1 + \tau_{13}ik)(1 + \tau_{23}jk)(1+\tau_{24}j\ell)(1+\tau_{14}i\ell), \nn 1/\text{const.}\!=2^4(1+\tau_{13} \tau_{14} \tau_{23}\tau_{24})\, , \]
where the $\tau$s  are now functions of the conditional odds-ratios given two further variables.\\[-4mm]

\noindent {\bf  \small  \textit To  obtain $\bm{\pi_{+jk\ell}}$}, one argues in the same way  as going  from equation \eqref{Joint123} to  \eqref{Marg23Tri}:
\begin{equation} \pi_{+jk\ell} =\text{const.}   (1 + \tau_{23}jk)(1 + \tau_{24}j\ell)(1 + \tau_{13}\tau_{14}k\ell), \nn  {1/\text{const.}}=2^3  (1+\tau_{13} \tau_{14} \tau_{23}\tau_{24}) \,.
\label{JointMarg23Cyc}  
\end{equation}
This shows (1) that the constant is just changed by a factor of two, (2) that the form of the distribution remains unchanged and (3) that by marginalising over 1, one gets $\tau_{34\setminus 1}=\tau_{13}\tau_{14}$.

By marginalising in the trivariate distribution of equation \eqref{JointMarg23Cyc} further over node 2, equation \eqref{InducedCorr4Cyc}    results so that the two dependence-inducing paths are traced
with hypetan interactions. 

One may also trace the contribution of the  $\tau$s for the marginal correlation of an edge present in the graph and obtains a similar effect as in equation \eqref{Marg23Tri}. For instance,
\begin{equation}\rho_{23}= \tau_{23 \setminus 1 4} =\text{const}\,(\tau_{23}+  \tau_{24}\tau_{14} \tau_{13}), \nn 1/\text{const.}= (1+\tau_{13} \tau_{14} \tau_{23}\tau_{24}) \,.
\label{Marg23Cyc}\end{equation}

These results extend directly to chordless cycles in more than four nodes which then prove  (1) that the 
marginal distributions of a chordless-cycle stays a cycle up to any trivariate marginal distribution and for  the marginal hypetan interaction  induced  in a bivariate distribution (2) that it results by tracing two paths if the edge is missing in the starting cycle and (3) that it coincides with the simple correlation if the  edge is present.\\[-2mm]

\noindent {\bf \small \textit{With an added independence for pair (2,3)}},  the discussed two cycles turn into 2-edge and 3-edge chains having $\tau$-parameters to each edge present in the graph which coincide with  simple correlations. Equations \eqref{Marg23Tri}, \eqref{Marg23Cyc} 
replicate then, respectively, the well-known result that the induced dependence for the endpoints in a linear Markov chain is the product of the simple correlations along the chain:
$   \tau_{23 \setminus 1}=\rho_{12} \rho_{13}$ and $ \tau_{23 \setminus 1 4} =  \rho_{24}\rho_{14} \rho_{13}. $

\subsubsection*{Appendix B:  Effects of conditioning}

In addition, when going to equations  \eqref{Marg23Tri},  \eqref{Marg23Cyc},
one gets  $\pi_{i|jk}=\pi_{ijk}/\pi_{+jk}$  and    $\pi_{i|jk\ell}$ as:
$$ \pi_{i|jk}= \fracshalf (1 + \tau_{12}ij)(1 + \tau_{13}ik) /   (1 + \tau_{12}\tau_{13} \,jk  ),   \nn
\pi_{i|jk\ell}=\fracshalf  (1 + \tau_{13}ik)(1+\tau_{14}i\ell)  / (1 + \tau_{13}\tau_{14}k\ell)   .$$
These have essentially the same form, stressing the similarity between the two types of cycles, the triangle and the chordless 4-cycle. 

 A standard factorisation of $\pi_{ijk}$ and Cochran's recursive relation for linear regression coefficients give a  linear parametrization of  the joint probabilities as 
  $$ \pi_{i|jk} \pi_{j|k} \pi_{k}= \fracseigth(1+\rho_{12}ij+\rho_{13}ik+\rho_{23}jk), \; \pi_{i|jk}=\fracshalf ( 1 + \beta_{1|2.3}ij +\beta_{1|3.2}ik), \;  \pi_{j|k}=\fracshalf ( 1 +   \rho_{23}jk).$$

 The expression for $ \pi_{i|jk}$ in terms of $\tau$'s simplifies, using  $j^2k^2=1$, by expanding the fraction with $(1 - \tau_{12}\tau_{13} \,jk  )$, to give:
  \begin{align*} \pi_{i|jk}&= \fracshalf (1+\tau_{12}ij+\tau_{13}ik+\tau_{12}\tau_{13}jk)(1- \tau_{12}\tau_{13} \,jk  )/(1-\tau_{12}^2\tau_{13}^2)\\
  &=\fracshalf [1+ \{\tau_{12}(1-\tau_{13}^2)/\text{const.}\}ij + \{\tau_{13}(1-\tau_{12}^2)/ \text{const.}\}ik ], \; \text{const.}=1-\tau_{12}^2\tau_{13}^2 \,,\end{align*}
  so that equation \eqref{BetaTauLam}  expresses the linear regression coefficients  in terms of  $\tau$s.

   By similar   arguments, one finds for the above chordless 4-cycle directly, for instance,
  $$\beta_{1|2.34}=0, \nn   \beta_{1|3.24} =\beta_{1|3.4}=\tau_{13}(1-\tau_{14}^2)/(1-\tau_{13}^2\tau_{14}^2)\, $$
 and  for the joint conditional distribution of $A_1, A_3$ given $A_4$, one gets
 $${\rm cov}(A_1,A_3| A_4=\ell)= \rho_{13}-\rho_{14}\rho_{34} \text{ since } \E(A_1A_3| A_4=\ell)=\rho_{13}, \; \E (A_i |A_4=\ell)=\rho_{i4}\ell.$$
Furthermore, since  also ${\rm var}(A_i|A_4=\ell)=1-\rho_{i4}^2$ does not vary with $\ell$ for $i=1,3$, the conditional correlations are constant at  all combinations  of $k,\ell$ and Theorem \ref{ConParCor} applies. 


\subsubsection*{Appendix C: Proof of Theorem 3}
{\small \bf Theorem 3}: {\it A quadratic exponential distribution for symmetric binary variables is generated over an edge-minimal hollow tree if and only if   (1) the
Markov structure  of its concentration  graph is defined by  the set of zeros in its   overall partial correlation matrix  and (2)  within each of its  prime graphs, all conditional  correlations agree  with  the partial correlations.}
\begin{proof}
Assume first that an Ising model is palindromic and its edge-minimal graph is a hollow tree. Then,  the joint distribution factorises according to its prime graphs  and cut-sets,
given any one of the proper node-set elimination schemes; see Theorem 1, Proposition 1 and equation \eqref{FactMarg}.  By definition, the prime graphs  of a hollow tree are edges or cycles. Cycles are either triangles or chordless  cycles  and  cut-sets contain at most two nodes.  

In the marginal distribution within each prime graph, all  conditional correlations coincide with the partial correlations; see  Theorem 2, Proposition 3, and equation \eqref{IndCov}.
In particular,  independences are captured  by zero log-linear 2-factor interactions and   by zero partial correlations. The factorisation of the joint density implies, together with the  independence constraints and the traceabilty  of linear relations,  that  a zero overall partial correlation equals the zero partial correlation within a prime graph. 

If  nodes $i,j$ belong to distinct prime graphs of a hollow tree, the smallest separating set, say  $S$,   such that 
 $X_i \ci X_j \mid X_S$ holds,  is a  cut-set of the graph. 
Then, the traceability implies for the conditional correlations,  
 $0=\rho_{ij \mid N\setminus \{ij\}} = \rho_{ij \mid S}$. As $S$ contains at most two  nodes, a conditional correlation 
equals also the  partial correlation, hence also $0=\rho_{ij.N\setminus \{ij\}}=\rho_{ij.S}$.  

For the  converse,  note first that  a graph without  hollow-tree structure has either $(i)$ a prime graph which is different from  an edge and from a cycle or $(ii) $ it has  a cut-set containing more than two nodes.  In case $(i)$, there is a prime graph  $P$ in which at least one node, say $i$, has more than two neighbours,   since    each  node has  exactly two neighbours   only in a  cycle. By definition,  no cut set and hence no outer node exists if   this prime graph $P$ is incomplete. Then,  there is  also no proper node-set elimination scheme  by which the distribution of the variables of $P$ could be generated via a  factorisation by which  the  logit main-effect regression with  more than two regressors, could be avoided.   Hence,  by Theorem 2 and Proposition 3, conditional and partial correlations involving variable $X_i$ do not agree within  $P$. 
 
In case $(ii)$, there are at least two prime graphs.  Let $i$ and $j$ be  nodes residing in different  prime graphs and having  cut-set $C$ containing more than two nodes. Then for the subset $\{i,j,C\}$ of the variables,  there is a factorisation  as $f_{ijC}=f_{i|C}f_{jC}$.  But by Proposition 3,   the main-effect logit regression  of $X_i$ on $X_C$ is  then not equivalent to a main-effect linear regression  of $X_i$ on $X_C$ with the  consequence  that  $\rho_{ij.C} \ne 0$ even though  $\rho_{ij|C}=0$.  
\end{proof}

 \subsubsection*{Appendix D: On traceabiliy}

A  distribution   is traceable if ($i$) its  graph is edge-minimal, ($ii$) its independences combine downwards and upwards
and ($iii$)  it is dependence-inducing; Wermuth (2012).

To ($i$): for an edge-minimal concentration graph,  every $ij$-edge present in the graph means
$A_i, A_j$ are dependent given all remaining variables since equation \eqref{DefDep} is satisfied.

To ($ii$): Combining independences  for sets $a,b,c,d$ partitioning $N'\subseteq N\setminus (a\cup b\cup c\cup d)$  applies to equivalences of $a\ci bd|c$. For palindromic Ising models with hollow-tree graphs:
$$(a\ci b|dc \text{ and } a\ci d|c) \iff (a\ci b|dc \text{ and } a\ci d|bc) \iff (a\ci b|c \text{ and } a\ci d|c).
$$
The implications of  $a\ci bd|c$ hold for all probability distributions, as well  as the converse of the first decomposition,  since  $f_{a|bdc}=f_{a|dc}$ together $f_{a|dc}=f_{a|c}$ implies the simplified factorisation $f_{a|bdc}=f_{a|c}$. The second  and third equivalence are not satisfied in general.  The downward and upward combinations result here with the recursive relation for concentrations and for covariances, respectively; see also Theorem 4 above.

To $(iii)$. A concentration graph model has been said to be association- or dependence-inducing if for each {\sf V} in the graph with outer nodes $i,j$, common neighbour $k$
and  $i\ci j|\{k,c\}$ given for some $c\subseteq N\setminus\{i,j,k,c\}$, the independence  changes to  $i\dep j|c$ by marginalising  over  the inner node $k$.
From Theorem \ref{IsGaussConc} and  the discussion of equation \eqref{BlockDiagComp}, the recursive relation among concentrations applies also to palindromic Ising models 
with hollow-tree structure. The vanishing of $\rho_{ij.N \setminus c}$ indicates conditional independence for all  $c$, which  -- based 
on a proper node-set elimination scheme -- do not introduce a dependence in the corresponding  marginal distribution; see Proposition \ref{MargPrimeG} and App. D.

Properties ($ii$) and ($iii$) have been attributed to Gaussian distributions by Ln\v{e}ni\v{c}ka. and Mat\'u\v{s} (2007). In some literature, distributions with properties ($ii$) have  been named compositional graphoids and those with property ($iii$)  singleton-transitive; see for instance Sadeghi and Wermuth (2016).  The dependence-inducing property ($iii$) applies to marginalising over single nodes but not 
to sets of variables; see the data to Figure \ref{DropAchievm}.

\subsubsection*{Appendix E: On partial inversion and partial closure}
{\bf  \textit{Partial inversion}} generalises the Beaton sweep operator, Dempster (1972, App.)  to non-symmetric matrices.   
Let   $N=(1, \ldots, d)$ denote the rows and columns of a real-valued matrix $ M$
having invertible leading principal submatrices.  For $a\subset N$ and $b=N \setminus  a$, one  may obtain $M^{-1}$ with two partial inversion steps:
$M^{-1}=  \inv_b (\inv_a M)$.

 The partial inversion operator is, with $a=\{1\}$:
\begin{equation} {M}=\begin{pmatrix}  s & { v}\T\\ {w}&  m\end{pmatrix} \nn  \n \inv_{\{1\}}\, {M}=\begin{pmatrix} 1/s  & -{v}\T/s\\ { w}/s& {m}-{w\,v}\T/s \end{pmatrix}. \label{pinvd}\end{equation}
 For $t>1$ elements in $a$,  $\inv_a   M$  applies this same  operation $t$ times,  using  permutations.
\begin{prop} {\rm  Wermuth, Wiedenbeck and Cox (2006).} The partial inversion operator  is commutative, can be undone and is exchangeable with taking 
submatrices.
\end{prop}
The first property shows for instance  that $ \inv_b (\inv_a M)= \inv_a (\inv_b M)$, the second  that $\inv_a\, M^{-1}= \inv_b \,M$ and with the  last, one derives equation \eqref{PcorRegcCon} with $N' \subset N$.\\

\noindent {\bf \textit{Partial closure}} is the analogue of the partial inversion operator  defined for edge matrices  $\mcal$. The operator derives  the structural zeros after partial inversion on Gaussian parameter matrices, that is it takes independences  into account but not any special parametric constellations.
Let   $N=(1, \ldots, d)$ denote the rows and columns of an  edge matrix $\mcal$,  zeros capture missing edges and ones edges of one type, such as full lines of a concentration graph in a unit symmetric matrix or arrows of a directed acyclic graph in a unit upper-triangular matrix. For $a\subset N$ and $b=N \setminus  a$, one obtains e.g.  the transitive  closure of a directed acyclic graph with two partial closure steps:
$\mcal^{-}=  \zer_b (\zer_a \mcal)$.

The partial closure operator is, with $a=\{1\}$ and  $`\In'$ denoting the indicator function:
\begin{equation}{ \mcal}= \begin{pmatrix}  1 & {\sf  v}\T\\ {\sf w}&  {\sf m}\end{pmatrix}: \nn  \n \zer_{\{1\}}\, { \mcal}=\begin{pmatrix} 1  & {\sf v}\T\\ {\sf  w}& \In[{\sf m}+{\sf{w\, v}}\T]\end{pmatrix}, \label{pinvd} \end{equation}
For $t\!>\!1\!$ elements in $a$,   $\zer_a \mcal$  applies this same  operation $t$ times,  using  permutations.
\begin{prop} {\rm  Wermuth, Wiedenbeck and Cox (2006).} The partial closure operator  is commutative, cannot be undone and is exchangeable with taking 
submatrices.
\end{prop}
The operator can be interpreted as tracing paths in the graph if the graph is edge-minimal. The exchangeability justifies,  for instance, repeated closure of {\sf V}s in subgraphs. Matrix versions of both operators have been explicitly exploited,  for instance, by Marchetti and Wermuth (2009, App.) and Wermuth (2011, 2012, 2015).

\subsubsection*{Appendix F: On Markov equivalence}

Starting from
any prime graph or a  cut-set, a hollow tree may be oriented -- based on  a proper
node-set elimination scheme  -- so that no sink {\sf V}, that is no $\snode \fra \snode \fla\snode $, and no sink 
 {\textsc U} is generated, that is
no $\snode \fra \snode \ful \snode \ldots\snode\ful\snode \fla\snode $. 
This avoids that any constraint $i\ci j| N\setminus\{i,j\}$ gets   the 
quite different meaning $i\ci j$ and applies the following  more general
result where these {\sf V-} and {\textsc U-}constellations  have been called `minimal complexes'.
\begin{thm} {\rm Frydenberg (1990).} Two LWF chain graphs, with identical sets of edges present, are Markov equivalent if and only if they have the same  
minimal complexes.
\end{thm}

Similarly, there is  a necessary and sufficient condition for the Markov equivalence of  other types of chain graph. Chain graph models with so-called regression graphs have sequences of joint responses and one last component  which is a concentration graph, hence has $i\ful j$.
Each component $i$ of a joint response may depend only on nodes $j$ in its past, drawn as $i\fla j$,   
and  the components within a given joint response  may be dependent given  their past, drawn as $i\dal j$; for these models see also Wermuth (2015). 

There are three types of  {\sf V}s 
which generalise the sink  {\sf V}  in a directed acyclic graph and  are called collision {\sf V}s:
 $ i\fra \snode \fla j, \n i\dal \snode \fla j,  \n  i\dal \snode \dal j.$ In  larger graphs, each of them states a conditional independence  of the uncoupled nodes excluding the inner node.
 All other {\sf V}s have been named transmitting  {\sf V}s since the are equivalent to  $i \fla \snode \fla j$.
 All transmitting {\sf V}s, including $i \ful \snode \ful  j$,  state conditional independence of the uncoupled nodes given the inner node and possibly other nodes conditioned upon.
 \begin{thm}{\rm Wermuth and Sadeghi (2012).} Two regression graphs, with identical sets of edges pesent, are Markov equivalent if and only if their sets with 
 collision  {\sf V\!}s coincide.
 \end{thm}
 Thus, an undirected graph is Markov-equivalent to a another regression graph, if and only if the latter has no collision {\sf V}. This holds by
 construction for every fattened tree. A distribution with a chordal fattened-tree structure is a covering model for the one with the hollow graph so
 that the independence structure of the former also  holds for the latter.

\subsection*{Acknowledgement} Computations were carried out using MATLAB; for drawing figures, we used CorelDraw and MATLAB. We thank Michael P. Weck for letting us reanalyze his data and David Cox  for most helpful discussions and comments, especially for  the proof of equivalence of the model formulations in terms of  the  log-linear and    the hypetan parameters.

\small
 \renewcommand{\baselinestretch}{1.3} 
\renewcommand\refname
 \normalsize


{\normalsize References}
 \begin{thebibliography}{100}
 


  \bibitem{Altham(1984)}
Altham, P (1984).
Improving the precision of estimation by
fitting a model.
 \textit{J.  Roy. Statist.  Soc.  B}, \textbf{46}, 118--119.
\vspace{-2mm}


\bibitem{BabaShibSib(2004)}
Baba, K., Shibata, R. and Sibuya M. (2004). 
\newblock Partial correlation and conditional correlation as measures of conditional independence.
 \textit{Austral. N. Zeal. J. Statist.}, \textbf{46},  657--664.
\vspace{-2mm}

\bibitem{BergsmaRudas(2002)}
Bergsma, W. and  Rudas, T. (2002). 
Marginal models for categorical data.
\textit{Ann. Statist.}, \textbf{30}, 140--159.
\vspace{-2mm}

\bibitem{Besag(1974)}
Besag, J. (1974).
Spatial interaction and the statistical analysis of lattice systems.
 \textit{J. Roy. Statist Soc. B}, \textbf{36}, 192--236. 
\vspace{-2mm}



\bibitem{Birch(1963)}
Birch, M.~W. (1963). 
 Maximum likelihood in three-way contingency tables.
 \textit{J. Roy. Statist. Soc. B}, \textbf{25}, 220--233.
\vspace{-2mm}

\bibitem{BishFienHol(1975)}
 Bishop, Y.~M.~M., Fienberg, S.~E., and Holland, P.~W. (1975).
  \textit{Discrete Multivariate Analysis: Theory and Practice}. MIT Press.
\vspace{-2mm}
\bibitem{Bolviken(1982)}
 Bolviken, E. (1982).
 Probability inequalities for the multivariate normal with non-negative
partial correlations, 
 \textit{Scand. J. Statist.}
 \textbf{9}, 49--58.
\vspace{-2mm}


\bibitem{CastRov(2006)}
Castelo, R. and Roverato, A. (2006). A robust procedure for Gaussian graphi\-cal model search from microarray data with p larger than n. \textit{J.  Machine Learn. Res.}, \textbf{7}, 2621--2650.
\vspace{-2mm}


\bibitem{Cox(2006)}
Cox, D.~R. (2006).
\textit{Principles of Statistical Inference}.
Cambridge: Cambridge Univ. Press.
\vspace {-2mm}


\bibitem{CoxWer)1990)}
Cox, D.~R. and Wermuth, N. (1990). An approximation to
maximum-likeli\-hood estimates in reduced models. {\em Biometrika}, 
{\bf 77}, 747--761.
\vspace{-2mm}

\bibitem{CoxWer(1994a)}
Cox, D.~R. and Wermuth, N. (1994a). 
A note on the quadratic exponential binary distribution.
\textit{Biometrika}, \textbf{81}, 403--408. 
\vspace{-2mm}



\bibitem{CoxWermuth(1994b} 
Cox, D.~R.  and  Wermuth, N. (1994b).
Tests of linearity, multivariate
normality and adequacy of linear scores. 
\textit{J. Roy. Statist. Soc. C}, \textbf{43}, 347--355.
  \vspace{-2mm}
  
  \bibitem{Dempster(1972)}
 Dempster, A.~P.
Covariance selection.
\textit{Biometrics}, \textbf{28},157--175.
\vspace{-2mm}



  

\bibitem{DahlVandToy2008}
Dahl, J., Vandenberghe, L. and  Roychowdhury, V. (2008).
Covariance selection for non-chordal graphs via chordal embedding. 
\textit{Optimiz. Meth. Softw.}, \textbf{23}, 501--520. 
\vspace{-2mm}

\bibitem{DarLauSp80}
 Darroch, J.~N., Lauritzen, S.~L., and Speed, T.~P.  (1980).
Markov fields and log-linear models for contingency tables.  \textit{Ann.  Statist.},    \textbf{8},  522--539.
\vspace{-2mm}


\bibitem{DarSpeed(1983)}
Darroch, J.~N. and Speed, T.~P. (1983). Additive and multiplicative models and interaction,
\textit{Ann. Statist.}, \textbf{11}, 724--738.
\vspace{-2mm}


  \bibitem{Dawid(1979)}
   Dawid, A.~P. (1979). Conditional independence in statistical theory,  \textit{J. Roy. Statist. Soc.  B}, \textbf{41}, 1--31.
\vspace{-2mm}


\bibitem{Dempster(1972)}
 Dempster, A.~P. (1972). Covariance selection. \textit{Biometrics}, \textbf{28}, 157--175.
 \vspace{-2mm}
 
 \bibitem{DetHos(2005)}
 Dethlefsen, C. and  H\o jsgaard, S. (2005). A common platform for graphical models in R: the gRbase package. \textit{J.  Statist. Softw.}
 \textbf{14}, No. 17.
 \vspace{-2mm}


\bibitem{Dinitz(2006)}
Dinitz, Y.  (2006).
Dinitz'  algorithm:  the original version and Even's  version. In: \textit{Theoret. Comput. Sci: Essays in Memory of Shimon Even}
 (eds.  S. Even, S.~O. Goldreich, A.~L. Rosenberg, and A.~L. Selman)
 Springer, New York,  218--240.
 \vspace{-2mm}
 
 \bibitem{Drton(2009)}
Drton, M.  (2009).
Discrete chain graph models. \textit{Bernoulli} \textbf{5}, 736--753.
 \vspace{-2mm}

\bibitem{Edwards(1963)}
Edwards, A.~W.~F.  (1963).
The measure of association in a 2 $\times$ 2 table.
\newblock\textit{J. Roy. Statist. Soc. Ser. A},
\textbf{126}, 109--114.
\vspace{-2mm}

 
\bibitem{Fallatetal(2016)}
Fallat, S., Lauritzen, S.,  Sadeghi, K.,  Uhler, C., 
 Wermuth, N. and Zwiernik P. (2016). 
 \newblock Total positivity in Markov structures. To appear in \textit{Ann. Statist.}, also on ArXiv:1510.01290.
 \vspace{-2mm}
 
  \bibitem{Fisher(1922)}
 Fisher, R.~A. (1922).
 \newblock On the mathematical foundations of theoretical statistics.
 \newblock \textit{Philos. Trans. Roy. Soc. London Ser. A},
  \textbf{222}, 309--368.
  \vspace{-2mm}

 
 
 \bibitem{Frydenb(1990)}
 Frydenberg, M. (1990). The chain graph Markov property.
 \textit{Scand.  J.
 Statist.}, \textbf{17}, 333--353.
\vspace{-2mm}


 \bibitem{Good(1958)}
 Good, I.~J. (1958). The interaction algorithm and practical Fourier analysis.
  \textit {J. Roy. Statist. Soc. B}, \textbf{20}, 361--372..
  \vspace{-2mm}
  
  

\bibitem{Goodman(1984)}
Goodman, L.~A. (1984). The Analysis of Cross-Classified Data Having Ordered Categories. Harvard Univ. Press, Cambridge.
 \vspace{-2mm}
 
\bibitem{Glonek(1996)}
Glonek,    G.~F.~V. (1996)
 A class of regression models for multivariate categorical responses.
\textit{Biometrika}, \textbf{83}, 15--28.
 \vspace{-2mm}
 
  
\bibitem{GlonekMcCul1995)}
Glonek,    G.~F.~V.  and McCullagh, P.  (1995).
 Multivariate logistic models.
 \textit {J. Roy. Statist. Soc. B}, \textbf{53}, 533--546.
  \vspace{-2mm}


\bibitem{Hagenaars(2015)}
Hagenaars, J.~A.   (2015). Methodological issues in categorical data analysis: categorization, linearity, and response effects. 
\textit{Methodology},  \textbf{11},  126--141.
  \vspace{-2mm}
  
  \bibitem{KarlinRinott(1983)}
 Karlin, S. and Rinott, Y. (1983). M-matrices as covariance matrices of multinormal distributions.
 \textit{Linear Algeb. Applic.} \textbf{52}, 419--438.  
   \vspace{-2mm}
   
\bibitem{LauUhlZwier(2017)} 
Lauritzen, S.~L., Uhler, C. and Zwiernik, P.  (2017). 
Maximum likelihood estimation in Gaussian models under total positivity.
 Submitted manuscript; on ArXiv: 1702.04031.
  \vspace{-2mm}
  
\bibitem{LauWer(1989)} 
Lauritzen, S.~L. and Wermuth, N. (1989). 
Graphical models for associations between variables, some of which are qualitative and some quantitative. 
\textit{Ann. Statist.} \textbf{17}, 31--54.
 \vspace{-2mm}



 \bibitem{Leimer(1993)} Leimer, H.-G. (1993). Optimal decomposition by clique separators. \textit{Discrete Mathematics}, \textbf{113}, 99--123.
 \vspace{-2mm}

\bibitem{LnenMatus(2007)}
Ln\v{e}ni\v{c}ka, R. and Mat\'u\v{s}, F.  (2007).
 On Gaussian conditional independence structures.
 \textit{Kybernetika}, \textbf{43},
323--342.
 \vspace{-2mm}

\bibitem{LohWain(2013)}
Loh, P. and Wainwright, M.~J. (2013).
\newblock Structure estimation for discrete graphical models: Generalized covariance matrices and their inverses.
\newblock \textit{Ann. Statist.}, \textbf{41},  3022--3049.
 \vspace{-2mm}
 
 
\bibitem{MarWer(2009)}
Marchetti, G.M. and Wermuth, N. (2009).
Matrix representations and independencies in
directed acyclic graphs.
\textit{Ann. Statist.}
\textbf{47}, 961--978.
 \vspace{-2mm}

\bibitem{MarchettiWer(2016)}
 Marchetti, G.~M. and Wermuth, N. (2016). Palindromic Bernoulli distributions.
\textit{Electron. J. Statist.}, 10,  2435--2460; also on ArXiv 1510.09072.
 \vspace{-2mm}
 
 

\bibitem{Matus(1984)}
Mat\'u\v{s}, F. (1994). On the maximum-entropy extensions of probability measures over undirected graphs. Proceed. 
\textit{Third Workshop on Uncertainty Processing in Expert Systems (WUPES'94), Inst. Information Theory Automation, Prague}, 181-198. 
 \vspace{-2mm}
 
 \bibitem{MitchHedetCockHedet(1981)}
Mitchell Hedetniemi, S.,  Cockayne, E.~J.  and Hedetniemi, S.~T.  (1981). Linear algorithms for finding the Jordan center and path center of a tree.
\textit{Transportation Science}, \textbf{15}, 98--114.
 \vspace{-2mm} 
 

 
 
\bibitem{Palmgren(1989)}
 Palmgren, J. (1989). Regression models for bivariate binary responses.  \textit{Biostatistics Working Paper Series.} 
 University of Washington.
  \vspace{-2mm}
  

\bibitem{RadRecZid(2017)} 
Radavi\v{c}ius, M. Reca\v{c}i\={u}s, T.  and \v{Z}idanavi\v{c}iut\'e, J.  (2017). Local symmetry of non-coding sequences for bacterial genoms.
\textit{Manuscript.}
 \vspace{-2mm}
 
\bibitem{Rose(1970)}
Rose, D. J.  (1970). Triangulated graphs  and the elimination process. \textit{J. Math Analysis Applic.}, \textbf{32}, 597--609.
 \vspace{-2mm}
 

\bibitem{SadMar(2012)}
Sadeghi K. and  Marchetti, G.~M. (2012).
Graphical Markov models with mixed graphs in R. 
\textit{The R Journal,}  \textbf{4}, 65--73.
\vspace{-2mm}
 
\bibitem{SadegWer16} 
Sadeghi, K. and  Wermuth, N. (2016). Pairwise Markov properties for regression graphs.
\textit{STAT}, \textbf{5}, 286--294, also on ArXiv1512.09016.
 \vspace{-2mm}
 
\bibitem{SpeedKiiv86} 
Speed, T.~P. and Kiiveri, H.~T. (1986). Gaussian Markov distributions over finite graphs.
 \textit{Ann. Statist}., \textbf{14}, 138--150.
 \vspace{-2mm}

\bibitem{StudenyBouk98}
 Studen\'y, M. and  Bouckaert, R.~R. (1998).
 On chain graph models for description of conditional independence structures.
\textit{Ann.  Statist.}, \textbf{8}, 1434--1495.
 \vspace{-2mm}

\bibitem{StudenyCuss(2016)}
 Studen\'y, M. and Cussens,  J. (2016). The chordal graph polytope  for learning decomposable models.
\textit{JMLR Workshop and Conference Proceedings}, \textbf{52},  499--510.
\vspace{-2mm}




\bibitem{TarjanYanna(1984)}
Tarjan, R.~E.  and  Yannakakis, M. (1984). Simple linear-time algorithms to test chordality of graphs, test acyclicity of hypergraphs, and selectively reduce acyclic hypergraphs.  \textit{SIAM J. Computing}, \textbf{13}, 566--579.
 \vspace{-2mm}

\bibitem{ThomasGreen(2009)}
Thomas, A.  and Green, P.~J. (2009). Enumerating the junction trees of a decomposable graph. \textit{J. Comp.  Graphical Statist.}, 18, 930--940.
 \vspace{-2mm}
 
  
 \bibitem{WagHalin(1962)} 
 Wagner, K. and Halin, R. (1962). Homomorhiebasen von Graphenmengen. \textit{Mathematische Annalen}. \textbf{147}, 126--142.
 \vspace{-2mm}
 
 \bibitem{Weck(1991)}
 Weck, M.~P. (1991).
 \textit{Der Studienfachwechsel. Eine L\"angsschnittanalyse der Interaktionsstruktur der Bedingungen des 
 Studienverlaufs.} Frankfurt: Lang.
  \vspace{-2mm}
 
 \bibitem{bibWer92} 
 Wermuth, N. (1992). On block-recursive regression equations (with discussion). \textit{Brazilian J. Probab. Statist.},  \textbf{6}, 1--56.
\vspace{-2mm}

\bibitem{Wer(1998)}
Wermuth, N. (1998). Pairwise independence.
{\em Encycl.  Biostat.},
P. Armitage and T. Colton (eds),  New York: Wiley, 3244--3246.
 \vspace{-2mm}
 

\bibitem{Wer(2011)}
Wermuth, N. (2011).
Probability models with summary graph structure.
\textit{Bernoulli}
\textbf{17},
845--879.
 \vspace{-2mm}

 
 
\bibitem{Wer(2012)}
Wermuth, N. (2012). 
\newblock Traceable regressions.
\newblock \textit{Int. Statist. Review}, \textbf{80},  415--438.
 \vspace{-2mm}
 
 
\bibitem{bibWerCox98} Wermuth, N. and Cox, D.R. (1998). On association models defined over independence graphs.  \textit{Bernoulli}
\textbf{4}, 477--495.
 \vspace{-2mm}


\bibitem{WerMar(2014)}
Wermuth, N. and Marchetti, G.~M. (2014). Star graphs induce tetrad correlations: for Gaussian
as well as for binary variables. \textit{Electron. J. Statist.}, \textbf{8}, 253--273. 
 \vspace{-2mm}




\bibitem{WerSadeg(2012)}
Wermuth N. and Sadeghi, K.  (2012).
Sequences of regressions and their independences (with discussion).
\textit{TEST},
\textbf{21}, 215--279. Also on ArXiv:1103.2523.
 \vspace{-2mm}
 


\bibitem{WerMarZwier(2014)} 
Wermuth, N., Marchetti, G.~M. and Zwiernik, P. (2014). 
Binary distributions of concentric rings. 
\textit{J.  Multiv. Analysis},
\textbf{130}, 252--260; also on ArXiv: 1311.5655
 \vspace{-2mm}



\bibitem{bibWerWieCox06}
Wermuth, N.,  Wiedenbeck, M. and Cox, D.R. (2006).
Partial inversion for linear systems and partial closure of independence
graphs.
\textit{BIT, Numerical Math.},
\textbf{46}, 883--901.
 \vspace{-2mm}


\bibitem{bibWieWer10}
Wiedenbeck, M. and Wermuth, N. (2010).
Changing parameters by partial mappings. \textit{Statist. Sinica},
\textbf{20},  823--836.
 \vspace{-2mm}


\bibitem{XieMaGeng(2008)}
Xie, X., Ma, Z. and Geng, Z. (2008). 
\newblock Some association measures and their collapsibility. 
\textit{Statist. Sinica},
\textbf{18}, 1165--1183.
\end{thebibliography}
\end{document}